\documentclass[a4paper,danish]{article}
\usepackage[margin=3cm]{geometry}

\usepackage[utf8]{inputenc}
\usepackage[T1]{fontenc}
\usepackage{natbib, dsfont, enumitem, authblk,
  amssymb,soul,xcolor,amsmath,amsthm,graphicx,subcaption,verbatim,pgfplots,tikz,prodint}
\usetikzlibrary{calc,patterns,angles,quotes,automata, positioning,arrows,shapes}

\usepackage{titlesec}
\titlespacing*{\section}{0em}{2em}{0em}
\titlespacing*{\subsection}{0em}{2em}{0em}
\titlespacing*{\subsubsection}{0em}{2em}{0em}

\setlist[itemize]{topsep=0pt}
\setlist[enumerate]{topsep=0pt}

\usepackage[author=]{fixme}
\fxusetheme{color}
\definecolor{fxtarget}{rgb}{.5,.5,.5}
\definecolor{fxnote}{rgb}{.5,.5,.5}

\definecolor{linkcolor}{rgb}{0, 0, 0.54}
\usepackage[colorlinks=true,allcolors=linkcolor,linktocpage=true]{hyperref} 

\newcommand{\R}{\mathbb{R}}
\newcommand{\N}{\mathbb{N}}
\newcommand{\blank}{\makebox[1ex]{\textbf{$\cdot$}}}

\renewcommand{\phi}{\varphi}
\renewcommand{\epsilon}{\varepsilon}
\newcommand*\diff{\mathop{}\!\mathrm{d}}

\newcommand\smallO{\textit{o}}
\newcommand\bigO{\textit{O}}
\newcommand{\midd}{\; \middle|\;}
\newcommand{\1}{\mathds{1}}
\usepackage{ifthen} %% Empirical process with default argument

\DeclareMathOperator*{\argmin}{\arg\!\min}

\newcommand{\arrow}[1]{\xrightarrow{\; {#1} \;}}
\newcommand{\arrowP}{\xrightarrow{\; P \;}} % convergence in probability

\renewcommand{\L}[2]{\ensuremath{\mathcal{L}^{{#1}}{({#2})}}}
\newcommand{\LP}{\ensuremath{\L{2}{P}}} % Shortcut for l_P^2 spaces

\newcommand{\empmeas}{\ensuremath{\mathbb{P}_n}} % empirical measure
\newcommand{\E}{{\ensuremath{\mathop{{\mathbb{E}}}}}} % expectation

\newcommand{\pd}[1]{\ensuremath{\frac{\partial }{\partial #1}
    \Big\vert_{#1=0}}}

\theoremstyle{plain} % plain italic style
\newtheorem{theorem}{Theorem}
\newtheorem*{theorem*}{Theorem}

\newtheorem{lemma}[theorem]{Lemma}
\newtheorem{proposition}[theorem]{Proposition}

\theoremstyle{definition} % non-italic

\newtheorem{assumption}{Assumption}
\newtheorem*{assumption*}{Assumption}

\theoremstyle{remark}

\newlist{assumptionlist}{enumerate}{1}
\setlist[assumptionlist,1]{
  label=(\thetheorem\alph*),
  leftmargin=*,
  topsep=0pt,
  align=left,
  labelsep=0.7em % space between label and text
}

\fxsetup{status=draft}

\title{Nonparametric inference for sublevel-set probabilities of
  conditional average treatment effect functions}
\author[1]{Anders Munch}
\author[1]{Thomas A.~Gerds}

\affil[1]{Section of Biostatistics, University of Copenhagen}

\date{\today}

\begin{document}
\maketitle

\begin{abstract}

  The average treatment effect can obscure important heterogeneity
  when individuals respond differently to a treatment. While the
  conditional average treatment effect (CATE) function captures such
  heterogeneity, it is difficult to communicate when it depends on
  many covariates. Sublevels sets of a multivariate CATE function are
  equally complicated objects, but the probability of a sublevel set
  of a CATE function is a single number with a simple interpretation
  as the proportion of individuals whose expected treatment effect
  does not exceed a prespecified threshold. By varying the threshold,
  a univariate monotone curve appears which can be used to visualize
  the overall type and degree of heterogeneity in a population. We
  formalize this curve as a target parameter and show that it is not
  pathwise differentiable under a nonparametric model. To address this
  nonstandard estimation problem, we leverage recent advances in
  monotone function estimation and develop a Grenander-type estimator
  that incorporates machine learning. We also show that the best
  piecewise linear approximation to the curve of interest is a
  pathwise differentiable parameter, and we develop a debiased machine
  learning estimator of this approximation. We investigate our
  proposed estimators' finite sample performance in a sequence of
  numerical studies based on data synthesized from a randomized trial.
  The methods are illustrated in data from a randomized trial on
  diabetes medication.
  
\end{abstract}

\section{Introduction}
\label{sec:introduction}

The average treatment effect provides an easily interpretable
univariate summary of a treatment's overall effect in a population,
but it leaves out important information when some subjects respond
differently to treatment. An active area of research within causal
inference is therefore estimation of conditional average treatment
effects and the identification of subgroups with differential
treatment effects. The conditional average treatment effect (CATE)
function provides a detailed picture of how a treatment works, but it
can be difficult to summarize and visual when many variables are
conditioned on. We propose to summarize the degree of treatment
heterogeneity by visualizing the proportion of the population that
belongs to the sublevel set of the CATE function defined by the level
$\alpha$, as a function of $\alpha \in \R$. Equivalently, this
function can be defined as the cumulative distribution function of the
expected treatment effects in the population. Hence this estimand is a
univariate, monotone function which is easy to visualize and
interpret.

To motivate and illustrate our method, we provide an example based on
data from a randomized trial investigating the cardiovascular effect
of liraglutide \citep{marso2016liraglutide}; details about the data
are given in Section~\ref{sec:application}. We first estimate the CATE
function using a causal random forest with default settings
\citep{wager2018estimation,tibshirani24}. This is done in two separate
folds of the data and the fitted causal forests are then used to
predict expected treatment effects conditional on baseline
characteristics in the remaining data that were not used to fit the
causal forest. Finally, we compute the empirical distribution function
of the aggregated estimated expected treatment effects. This is a
version of the cross-fitted plug-in estimator defined formally in
Section~\ref{sec:plug-estimator}. The resulting estimator for this
data set is shown in the left panel of
Figure~\ref{fig:grf-plugin-leader} along with the average treatment
effect. We see that while the average treatment effect is positive,
the curve suggests that for around 35\% of the trial population the
treatment is not expected to be beneficial, while around 50\% has a
larger expected effect of the treatment than the average patient in
the trial. The right panel of Figure~\ref{fig:grf-plugin-leader} shows
a histogram of the estimated expected treatment effects, which also
indicate that some treatment heterogeneity is present in the trial
population.

\begin{figure}
  \centerline{\includegraphics[width=1\linewidth]{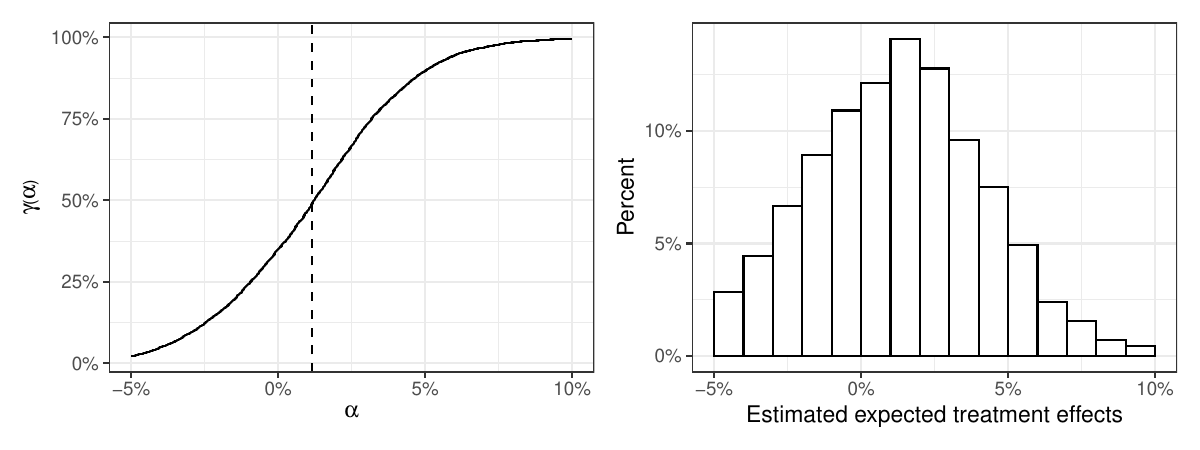}}
  \caption[]{Visualization of the estimated expected treatment effects
    based on data from a randomized trial. The left panel show the
    estimated cumulative distribution function of the expected
    treatment effects estimated by a causal random forest. The dashed
    line is the estimated average treatment effect obtained from the
    same causal random forest. The right panel is a histogram of the
    distribution of the estimated expected treatment effects.}
  \label{fig:grf-plugin-leader}
\end{figure}

In this article, we formally introduce the curve estimated in the left
panel of Figure~\ref{fig:grf-plugin-leader} as an estimand of interest
and illustrate some of its properties. We then examine the statistical
complexity of estimating this curve, showing in particular that it is
not a pathwise differentiable parameter. In addition to a simply
plug-in strategy, we propose two estimators relying on modern tools
from monotone curve estimation theory and debiased machine
learning. Both these estimation strategies rely on estimators of the
antiderivative(s) of the function of interest, and a key step is
showing that the pointwise evaluation of these antiderivatives are
pathwise differentiable and finding their efficient influence
functions.

The present work relies on the rich literature concerned with
estimation of the CATE function
\citep[e.g,][]{hill2011bayesian,lu2018estimating,kunzel2019metalearners,wager2018estimation,nie2021quasi,kennedy2023towards}. In
particular, our proposal is related to recent works that suggest
targeting low-dimensional features or projections of the CATE
function. For instance, \cite{semenova2021debiased} and
\cite{chernozhukovFisher} introduce estimation frameworks that target
summaries and approximations of the CATE function, and
\cite{levy2021fundamental,hines2022variable}, and
\cite{ziersen2025variable} define variable importance measures that
can be phrased as univariate summaries of the CATE function.

A closely related concept to the function-valued parameter of interest
suggested in the present paper is the sorted effects method introduced
by \cite{chernozhukov2018sorted}. The sorted effects method can be
expressed as estimating the inverse function of our parameter of
interest, i.e., the quantile function of the expected treatment
effects. \cite{chernozhukov2018sorted} consider parametric and
semi-parametric models and provide results that allow for the
construction of uniform confidence bands based on having an expression
for the asymptotic distribution of the CATE estimator used. In
contrast, we consider a fully nonparametric model and construct
estimators of the parameter of interest directly using tools from
monotone curve estimation theory. Our suggested estimators can employ
nonparametric data-adaptive methods such as machine learning to
estimate the CATE function. Allowing for flexible model estimation is
particularly important for the present purpose, because
(semi-)parametric models typically impose particular restrictions on
the type and degree of heterogeneity, e.g., through the number of
interaction terms and choice of link function in a generalized linear
model.

Another related concept is the targeting operator characteristic (TOC)
defined by \cite{yadlowsky2025evaluating}. For some treatment
prioritization rule \( S \colon \R^d \rightarrow \R \) and level
$u \in (0,1]$, the number \( \mathrm{TOC}(u; S) \) is the difference
between the average treatment effect in the population and the
treatment effect among subjects for which
\( S(W) \geq F_S^{-1}(1-u)\), where \( F_S \) denotes the cumulative
distribution function of \( S(W) \), and \( W \in \R^ d \) is a vector
of baseline covariates on which the prioritization of treatment is
based. \cite{yadlowsky2025evaluating} use various weighted averages of
\( u \mapsto \mathrm{TOC}(u; S) \) to quantify the performance of some
pre-specified prioritization rule \( S \). Compared to their proposal,
we consider a setting where the prioritization rule \( S \) is the
actual CATE function and has to be estimated from data. Also, our
effect measure is quantified in terms of the volume of the
sublevel sets (as measured by the marginal probability), instead of in
terms of the treatment effect within these sublevel sets.

The problem of estimating sublevel sets of densities and regression
functions has a long history
\citep[e.g.,][]{polonik1995measuring,tsybakov1997nonparametric,cavalier1997nonparametric,rigollet2009optimal,chen2017density}.
Recently, \cite{reeve2023optimal} and \cite{bonvini2023minimax}
considered estimation of sublevel sets of the CATE function. Our work
diverges from these by treating the sublevel set as a nuisance
parameter and letting instead its probability volume be the quantity
measure of interest. Also, we consider estimation for all threshold
values.

We introduce our notation and estimand in
Section~\ref{sec:notation-estimand}. In
Section~\ref{sec:non-pathw-diff}, we demonstrate that the estimand is
a non-standard estimation problem by proving that it is not a pathwise
differentiable parameter. In Section~\ref{sec:estim-thro-antid}, we
propose three different estimators and consider their asymptotic
properties. We consider a naive plug-in estimator, a Grenander-type
estimator, and an estimator targeting a first order spline
approximation of our estimand of
interest. Section~\ref{sec:numer-exper} contains the results of our
numerical experiments, while Section~\ref{sec:application} demonstrate
the results we obtain by employing our proposed estimators to the
liraglutide study. Section~\ref{sec:discussion} concludes with a
discussion of our findings. All proofs are deferred to the Appendix.

\section{Notation and estimand}
\label{sec:notation-estimand}

Following the Neyman-Rubin potential outcomes framework
\citep{neyman1923applications,rubin1974estimating,hernanRobinsWhatIf}, let \( Y^a \) denote the
potential outcome of \( Y \) under treatment \( A=a \). We assume that
\( Y \) and \( A \) are binary, and we let
\( W \in \mathcal{W} \subset \R^d \) denote an observed vector of
baseline covariates. Let \( (Y^1, Y^0, W) \sim Q \in \mathcal{Q} \),
where \( \mathcal{Q} \) denotes the collection of all probability
measures on \( \{0,1\}^2 \times \mathcal{W} \). The conditional
average treatment effect (CATE) function is the function
\( w \mapsto \E_Q{ \left[ Y^1 - Y^0 \midd W=w \right]}\), for
\( w \in \mathcal{W} \), where we use \( \E_{Q} \) to denote
expectation under the law \( Q \).

Instead of samples from some \( Q \), we observe samples of the form
\( O = (W, A, Y) \), where \( Y=A \, Y^1 + (1-A)\, Y^0 \) and the
conditional distribution of \( A \) given \( (Y^1, Y^0, W) \) is
determined by some conditional probability \( G \in \mathcal{G} \),
where \( \mathcal{G} \) denotes a collection of conditional
probability distributions. Any distribution for \( O \) is determined
by a \( Q \in \mathcal{Q} \) and \( G \in \mathcal{G} \), so we write
\( \mathcal{P}_{\mathcal{Q}, \mathcal{G}} \) for the model for the
observed data implied by \( \mathcal{Q} \) and \( \mathcal{G} \), and
we use \( P_{Q,G} \in \mathcal{P}_{\mathcal{Q}, \mathcal{G}} \) to
denote a distribution on the sample space
\( \mathcal{O} = \mathcal{W} \times \{0,1\} \times \{0,1\}  \) of the
observed data determined by \( Q \) and \( G \).  Assuming coarsening
at random (or conditional exchangeability) and positivity for the
family \( \mathcal{G} \), the CATE function is identifiable from any
distribution \( P_{Q,G} \in \mathcal{P}_{\mathcal{Q}, \mathcal{G}} \)
through the formula
\begin{equation*}
    \E_Q{ \left[ Y^1 - Y^0 \midd W=w  \right]} = \tau(P_{Q,G})(w),
  \quad \forall w \in \mathcal{W},
\end{equation*}
where
\begin{equation*}
  \tau(P)(w)  =
  \E_P{ \left[ Y \midd A=1, W=w  \right]}
  - \E_P{ \left[ Y \midd A=0, W=w  \right]}  ,
\end{equation*}
for all \( P \in \mathcal{P}_{\mathcal{Q}, \mathcal{G}} \)
\citep{robins1986new,gill1997coarsening,van2003unified}. 

We let \( \mathcal{P} \) denote a model for the observed data, and
introduce the following function-valued nuisance parameters defined on
\( \mathcal{P} \):
\begin{align}
  \label{eq:10}
  \pi(P)(w) & = P(A=1 \mid W=w),
  \\
  \label{eq:11}
  \mu(P)(a,w) & = \E_P{\left[ Y \mid A=a, W=w  \right]}.
\end{align}
Note that
\begin{equation}
  \label{eq:28}
  \tau(P)(w) = \mu(P)(1, w) - \mu(P)(0, w) .
\end{equation}
When the model for the counterfactual data \( \mathcal{Q} \) is
unspecified, coarsening at random entails no structural restrictions
on the observed data model
\( \mathcal{P}_{\mathcal{Q}, \mathcal{G}} \). The positivity
assumption implies a bound on $\pi$, which we make explicit with the
following stronger uniform positivity assumption.

\begin{assumption}
  \label{assum1}
  There exists a constant \( c >0 \) such that for all
  \( P \in \mathcal{P} \) and \( w \in \mathcal{W} \), it holds that
  \( c \leq \pi(P)(w) \leq 1-c \).
\end{assumption}

Assumption~\ref{assum1} implies that $\mu$ is well-defined. Define for
any \( P \in \mathcal{P} \), where \( \mathcal{P} \) fulfills
Assumption~\ref{assum1}, and \( \alpha \in [-1,1] \) the real-valued
parameter
\begin{equation}
  \label{eq:27}
  \gamma(P)(\alpha) =
  \int_{\mathcal{W}} \1{\{\tau(P)(w) \leq \alpha\}}  P(\diff w)
  % =
  % P{(\{ w \in \mathcal{W} : \tau(P)(w) \leq \alpha \})}
  =
  P(\tau(P)(W) \leq \alpha).
\end{equation}
Let \( \mathcal{D}_{[-1,1]} \) denote the space of càdlàg functions
with domain \( [-1,1] \). Our main parameter of interest is the
function-valued parameter
$\gamma \colon \mathcal{P} \rightarrow \mathcal{D}_{[-1,1]}$, where
the function \( \gamma(P) \colon [-1,1] \rightarrow \R \) is defined
as \( \alpha \mapsto \gamma(P)(\alpha) \). We refer to $\gamma$ as the
sublevel function.

To ease notation, we shall sometimes in the following suppress the
dependence on \( P \). For instance, we will sometimes write
$\mu(a,w)$ and $\gamma(\alpha)$, with the implicit understanding that
both values depend on \( P \), and we shall also write, e.g.,
$\gamma(\alpha) \colon \mathcal{P} \rightarrow \R$, which is formally
to be understood as the map
$\gamma(\blank)(\alpha) \colon \mathcal{P} \rightarrow \R$.

To see how treatment heterogeneity can be read off from the curve
$\gamma$, we visualize its shape for a simple data-generating
distribution under various parameter settings. We consider the
data-generating distribution determined by the relations
\begin{equation}
  \label{eq:simple-dgm}
  \begin{split}
    W_1 & \sim \mathrm{Unif}{([-1,1])},
          \\
    W_2&  \sim \mathrm{Bernoulli}(p) ,
    \\
    \E{\left[ Y \mid A, W_1, W_2 \right]}
        &
    = \mathrm{expit}
    \left(
    \beta_1 W_1 A  +  \beta_2 W_2 A 
    \right),
  \end{split}
\end{equation}
where \( A \in \{0,1\} \) indicates treatment. The covariates
\( W_1 \) and \( W_2 \) can be interpreted as, e.g., a continuous a
biomarker and gender, respectively. The parameters $\beta_1$ and
$\beta_2$ control the degree of interaction between the treatment and
each of the two covariates. In Figure~\ref{fig:illu-gamma}, we have
drawn the sublevel function \( \gamma \) defined in
equation~\eqref{eq:27} for different values of \( p \), \( \beta_1 \),
and \( \beta_2 \). The figure demonstrates that when there are no
interactions (\( \beta_1 = 0 \) and \( \beta_2 = 0 \)), \( \gamma \)
is a degenerate cumulative distribution with all point mass at the
average treatment effect. Introducing an interaction between treatment
and a continuous covariate (\( \beta_1 >0 \) and \( \beta_2 =0\))
makes $\gamma$ a continuous function. When there is only an
interaction between treatment and a discrete covariate
(\( \beta_1 =0 \) and \( \beta_2 =1\)), $\gamma$ has a plateau which
illustrates that there are two groups with differential treatment
effects. When there are interactions between both a continuous and a
discrete covariates (\( \beta_1 >0 \) and \( \beta_2 =1\)), this
plateau is smoothed out to some degree depending on the size of the
interaction terms. Finally, we see that when an interaction term
between treatment and \( W_2 \) is present, the shape of \( \gamma \)
also depends on the prevalence of \( W_2 \).

\begin{figure}
  \centerline{\includegraphics[width=1\linewidth]{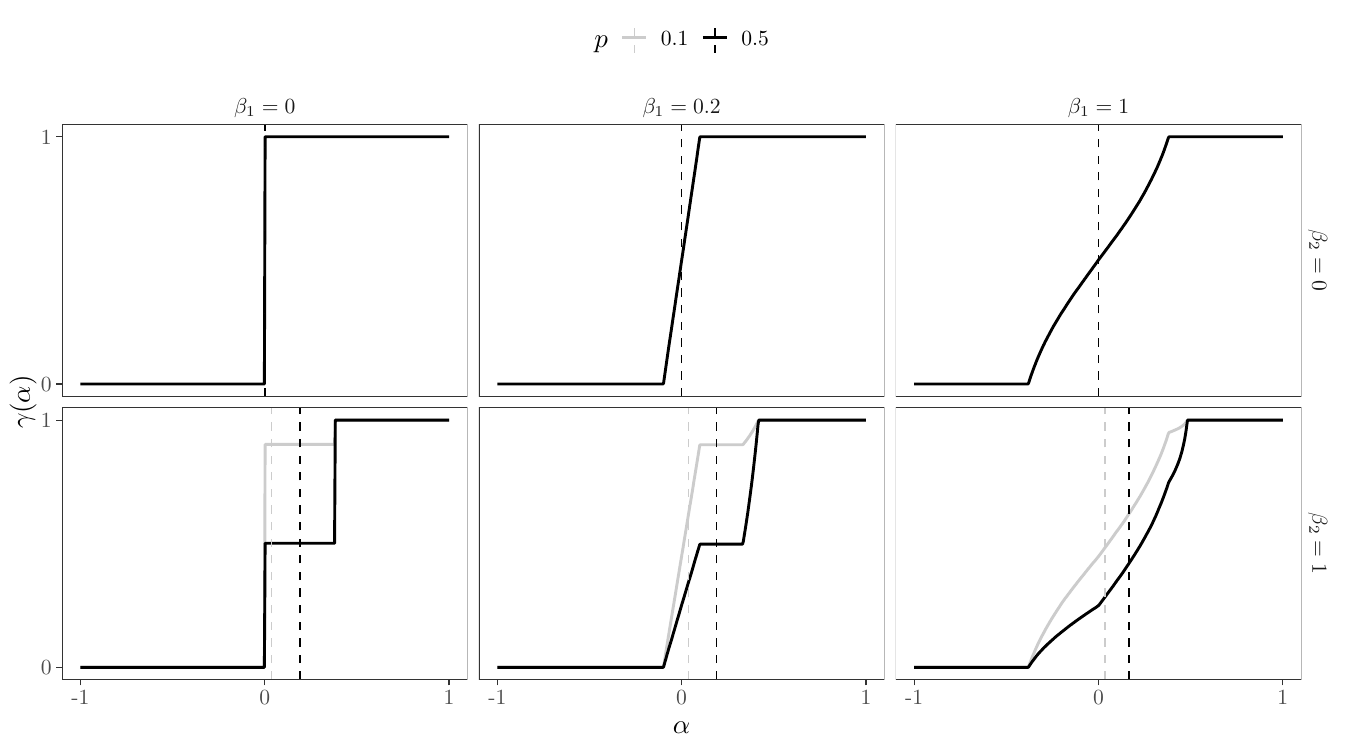}}
  \caption[]{The solid curve is the sublevel function \( \gamma \)
    defined in equation~\eqref{eq:27} under different data-generating
    distributions determined by the values \( p \), \( \beta_1 \), and
    \( \beta_2 \), c.f., equation~\eqref{eq:simple-dgm}. The dashed
    lines indicate the average treatment effect.}
  \label{fig:illu-gamma}
\end{figure}

Besides providing a visual summary of the degree and type of
heterogeneity in the population, each value $\gamma(\alpha)$ has a
clear interpretation as the proportion of the population with an
expected treatment effect below or equal to $\alpha$. For instance,
$\gamma(0)$ denotes the proportion of the population for which the
treatment effect is not expected to be beneficial. We emphasize that
this interpretation is relative to which set of conditioning
covariates is used to define the CATE function, as this defines how
fine or coarse the expected treatment effects are defined. Hence,
conditioning on different sets of covariates in the definition of the
CATE function can lead to different associated sublevel functions. An
attractive feature of a sublevel function $\gamma$ is that it will
always be a univariate function, no matter the dimension of the
covariates space \( \mathcal{W} \).

\section{A non-standard estimation problem}
\label{sec:non-pathw-diff}

Before constructing estimators of the sublevel function, it is of
interest to understand how difficult the estimation problem is. A
first step in this direction is to clarify whether $\gamma(\alpha)$ --
i.e., the pointwise evaluation of the sublevel function -- is a
pathwise differentiable parameter. The precise definition of this
condition is given in Appendix~\ref{sec:gener-result-pathw}. Pathwise
differentiability is a necessary condition for the existence of
regular asymptotically linear estimators
\citep{van1991differentiable}. In this section, we show that
$\gamma(\alpha)$ is not in general pathwise differentiable. This
implies that the usual influence function-based tools from
semiparametric efficiency theory is not immediately available
\citep{bickel1993efficient,van2000asymptotic,van2011targeted,kennedy2022semiparametric}.

To formally state the result of this section, we define the model
\( \mathcal{P}_S \) as follows. The model \( \mathcal{P}_S \) fulfills
Assumption~\ref{assum1}, and for all
\( P \in \mathcal{P}_{\mathcal{S}} \) the CATE function \( \tau(P) \)
is smooth (infinitely many times differentiable) and the marginal
distribution of \( W \) is dominated by Lebesgue measure.

\begin{theorem}
  \label{theorem:target-nondiff}
  Fix \( \alpha \in (-1,1) \) and \( P \in \mathcal{P}_{S} \), let
  $\nabla\tau(P)$ denote the gradient of \( w \mapsto \tau(P)(w) \),
  and let \( \tau(P)^{-1}(\alpha) \subset \mathcal{W} \) denote the
  $\alpha$-level-set of $\tau(P)$.  If either
  $\nabla\tau(P)(w) \not = 0$ for some
  \( w \in \tau(P)^{-1}(\alpha) \) or $\nabla\tau(P)(\mathcal{U}) = 0$
  for some open subset $\mathcal{U} \subset \mathcal{W}$ with
  $\alpha \in \tau(P)(\mathcal{U})$, then
  \( \gamma(\alpha) \colon \mathcal{P}_S \rightarrow \R \) is not
  pathwise differentiable at \( P \).
\end{theorem}

The two cases in Theorem~\ref{theorem:target-nondiff} cover the
situation where the \( \alpha \)-level-set of the CATE function
contains a non-critical point, and the situation where the CATE
function is locally constant. A situation not covered by the theorem
is, for instance, values at which the CATE function has a
saddle-point. We conjecture that pathwise differentiability will also
fail in this case.

To gain some intuition for the why $\gamma(\alpha)$ fails to be
pathwise differentiable, consider the special case where
\( W \sim \mathrm{Unif}([0,1]) \) and the CATE function \( \tau \) is
bijective. In this case, \( \gamma(\alpha) = \tau^{-1}(\alpha) \),
where $\tau^{-1}$ denotes the inverse of $\tau$. The pointwise
evaluation of the CATE function is not pathwise differentiable in a
non-parametric model, and hence we would not expect its inverse to be
pathwise differentiable either.

The model \( \mathcal{P}_S \) imposes a high-degree of smoothness and
will thus be a submodel of most non-parametric models. This means that
Theorem~\ref{theorem:target-nondiff} implies that pathwise
differentiability also fails under more general non-parametric
models. On the other hand, estimation of $\gamma(\alpha)$ can be
related to a classification problem with decision set
\( \{ w : \tau(w) \leq \alpha\} \), and it is known that if \( \tau \)
fulfills a certain margin condition in the vicinity of the level
$\alpha$, then fast rates of convergence for the this classification
problem are achievable \citep{mammen1999smooth,audibert2007fast}. A
discontinuous CATE function can fulfill this margin condition, which
corresponds to a situation in which the sublevel function $\gamma$ is
locally flat. It is thus conceivable that the parameter
\( P \mapsto \gamma(P)(\alpha) \) is pathwise differentiable for
certain levels $\alpha$ under a model \( \mathcal{P} \) that assumes
$\tau$ can only take on a finite number of values, but we do not
investigate this further here.

\section{Estimators of the sublevel function}
\label{sec:estim-thro-antid}

In this section we propose three different estimators of $\gamma$
which are based on an iid.\ data set \( \{O_i\}_{i=1}^{n} \), with
\( O_i \sim P \). All of them rely on estimators of sublevel sets. The
sublevel set at level $\alpha$ is uniquely characterized by the
indicator function
\begin{equation}
  \label{eq:12}
  \eta_{\alpha}(P)(w) = \1{\{\tau(P)(w) \leq \alpha\}}.
\end{equation}
An estimator of $\eta_{\alpha}$ can be obtained by plugging an
estimator of $\tau$ into equation~\eqref{eq:12}, while an estimator of
$\tau$ can in turn be obtained by plugging an estimator of $\mu$ into
equation~\eqref{eq:28}. We shall assume that such plug-in estimators
of $\eta_{\alpha}$ and $\tau$, obtained from a given estimator
$\hat{\mu}_n$, are used. We discuss other strategies for estimating
\( \eta_{\alpha} \) and $\tau$ in Section~\ref{sec:discussion}. Some
of our proposed estimators in this section also rely on an estimator
\( \hat{\pi}_n \) of the propensity score $\pi$. We use $\hat{P}_n$ as
generic notation for a tuple of estimators
\( \hat{P}_n = (\hat{\mu}_n, \hat{\pi}_n, \bar{\mathbb{P}}_n) \),
where \( \bar{\mathbb{P}}_n \) is used to denote the marginal
empirical measure of the covariates \( \{W_i\}_{i=1}^n \). For all
proposed estimators of $\gamma$ we employ cross-fitting: For some
(small) \( K \in \N \), e.g., \( K=2 \) or \( K=5 \), we partition the
data set independently into \( K \) folds and let \( n_k \) denote the
number of observations in each fold. We let \( \mathbb{P}_n^k \)
denote the empirical measure of the \( k \)'th fold and use
\( \hat{P}_n^{-k} \) to denote an estimator fitted on all data except
the \( k \)'th fold; e.g., \( \hat{\tau}_n^{-k} \) denotes the
estimator $\hat{\tau}_n$ fitted on all data except the \( k \)'th
fold. We make the following mild assumption on the partition of the
data.

\begin{assumption}
  \label{assum2}
  It holds that \( n/n_k = K + \smallO_P{(1)} \).
\end{assumption}

Finally, we make the assumption that the propensity score estimator is
uniformly bounded away from 0 and 1.

\begin{assumption}
  \label{assum3}
  There exists a constant \( c >0 \) such that
  \( c \leq \hat{\pi}_n(w) \leq 1-c \) for all
  \( w \in \mathcal{W} \).
\end{assumption}

\subsection{Cross-fitted plug-in estimator}
\label{sec:plug-estimator}

The plug-in estimation strategy is based on inserting estimated values
into the definition of the sublevel function $\gamma$ given in
equation~\eqref{eq:27}. This is based on using an estimator
$\hat{\tau}_n$ of the CATE function to obtain the expected treatment
effects of all samples $\hat{\tau}_n(W_i)$, \( i=1, \dots, n \). We
then use the empirical distribution function of the estimated expected
treatment effects \( \{\hat{\tau}_n(W_i)\}_{i=1}^n \) as an estimator
of $\gamma$.  To avoid overfitting, we propose a cross-fitted version
of this estimator, which is defined as
\begin{equation}
  \label{eq:gamma-plug-in}
  \hat{\gamma}_n^{\text{pl}}(\alpha) = \frac{1}{K}\sum_{k=1}^{K}\mathbb{P}_n^k{
    [\hat{\eta}_{\alpha,n}^{-k}]},
  \quad \text{where} \quad
  \hat{\eta}_{\alpha,n}^{-k}(w) = \1{\{\hat{\tau}_n^{-k}(w) \leq \alpha\}},
  \quad
  \alpha \in [-1,1].
\end{equation}

Plug-in estimators are easy to use but it is often difficult to derive
their asymptotic distribution when data-adaptive methods like machine
learning algorithms are used. In our case, this can be seen from the
decomposition
\begin{equation*}
  \hat{\gamma}_n^{\text{pl}}(\alpha) - \gamma(\alpha)
  = \frac{1}{K}\sum_{k=1}^{K} 
  P{[\hat{\eta}_{\alpha,n}^{-k} - \eta_{\alpha}]}
  + \frac{1}{K}\sum_{k=1}^{K} 
  (\mathbb{P}_n^k- P){[\hat{\eta}_{\alpha,n}^{-k}]},
\end{equation*}
where we had added and subtracted
\( P{[\hat{\eta}_{\alpha,n}^{-k}]} \) to the difference
\( \hat{\gamma}_n^{\text{pl}}(\alpha) - \gamma(\alpha) \). As
\( \mathbb{P}_n^k \) and \( \hat{\eta}_{\alpha,n}^{-k} \) are
independent, standard arguments can be used to show that the second
term in the display converges at the parametric rate to a centered
Gaussian distribution with variance
\( \gamma(\alpha)(1-\gamma(\alpha)) \) if $\hat{\eta}_{\alpha,n}$ is a
consistent estimator of $\eta_{\alpha}$
\cite[e.g.,][]{chernozhukov2018double,kennedy2022semiparametric}. However,
\( \hat{\eta}_{\alpha,n}^{-k} \) itself will typically converge at a
slower than parametric rate, and hence the first term in the display
above will be dominating. The summands of the dominating term can be
bounded by
\begin{equation*}
  P{[|\hat{\eta}_{\alpha,n}^{-k} - \eta_{\alpha}|]}
  = P(\hat{\eta}_{\alpha,n}^{-k}(W) \not = \eta_{\alpha}(W)),
\end{equation*}
which is referred to as the disagreement probability or the
probability of symmetric difference between $\hat{\eta}_{\alpha,n}$
and $\eta_{\alpha}$. Under a smoothness condition on $\gamma$ at
$\alpha$, the disagreement probability can be bounded by the
supremum-norm of the difference between \( \hat{\tau}_n^{-k} \) and
\( \tau \) \citep{audibert2007fast,rigollet2009optimal}. We may thus
expect that the plug-in estimator \( \hat{\gamma}_n^{\text{pl}} \)
inherits the sup-norm convergence rate of the chosen CATE learner.

There exist several suggestions for the construction of confidence
regions using kernel-based plug-in estimators of regression functions
and densities
\citep{mason2009asymptotic,mammen2013confidence,chen2017density,qiao2019nonparametric}. Recently,
\cite{bonvini2023minimax} proposed a method for constructing
asymptotically valid confidence regions for sublevel sets of the CATE
function based on uniform confidence bands for the CATE function
itself. \cite{chernozhukov2018sorted} showed that the operator that
maps a CATE function to the quantile function $\gamma^{-1}$ is
Hadamard differentiable under some regularity conditions, and used
this to establish asymptotic results for estimators of $\gamma^{-1}$
based on the asymptotic distribution of the plugged in CATE
learner. These results could potentially be adapted to construct
confidence intervals for the plug-in estimator defined in
equation~\eqref{eq:gamma-plug-in}. In the following two subsections we
instead attempt to construct alternative estimators of $\gamma$
directly with tractable asymptotic distributions that do not depend on
the asymptotic distribution of the specific choice of CATE learner but
only on its rate of convergence.

\subsection{Grenander-type estimator}
\label{sec:gren-type-estim}

\cite{westling2020unified}, using ideas from
\cite{groeneboom1983density} and \cite{van2006estimating}, provide
asymptotic theory for so-called generalized Grenander-type
estimators. A Grenander-type estimator of a nondecreasing function is
based on the fact that the antiderivative of a nondecreasing function
is convex. This leads to the idea of estimating a nondecreasing
function of interest with the derivative of the greatest convex
minorant of an estimator of the function of interest's
antiderivative. As \( \gamma \) is a cumulative distribution function
and hence monotone, we can employ a Grenander-type estimator to
estimate $\gamma$.

We use $\Gamma$ to denote the antiderivative of $\gamma$, i.e., we
define the function-valued parameter $\Gamma \colon \mathcal{P}
\rightarrow \mathcal{D}_{[-1,1]}$ as
\begin{equation*}
  \Gamma(P)(\alpha) = \int_{-1}^{\alpha} \gamma(P)(u)  \diff u,
  \quad \alpha \in [-1,1].
\end{equation*}
The advantage of lifting the estimation problem from $\gamma$ to
$\Gamma$ is that estimation of an antiderivative is often easier than
estimation of the function itself. Indeed,
Theorem~\ref{theorem:anti-target-diff} demonstrates that
$\Gamma(\alpha)$ is a pathwise differentiable parameter if $\gamma$ is
continuous at $\alpha$. In the following, we use the notion of a
saturated tangent space, which means that the tangent space is the
collection of all \( P \)-zero mean functions in \( \LP \). This can
be taken as a definition of what it means for a model
\( \mathcal{P} \) to be fully nonparametric. We refer to
Appendix~\ref{sec:gener-result-pathw} for the precise definition of a
tangent space.

\begin{theorem}
  \label{theorem:anti-target-diff}
  Let \( \mathcal{P} \) be a model with saturated tangent space and
  fulfilling Assumption~\ref{assum1}, and let \( \alpha \in (-1, 1) \)
  be fixed. The parameter \( \Gamma({\alpha}) \) is pathwise
  differentiable at any \( P \in \mathcal{P} \) for which
  \( u \mapsto \gamma(P)(u) \) is continuous at $\alpha$. In this
  case, the efficient influence function of \( \Gamma({\alpha}) \) is
    \begin{equation}
    \label{eq:anti-target-can-grad}
    \upsilon_{\alpha}(P)(O) = \1{\{\tau(P)(W) \leq
      \alpha\}}(\alpha- \phi(P)(O)) - \Gamma(P)({\alpha}),
  \end{equation}
  where
  \begin{equation*}
    \phi(P)(O) =
    \tau(P)(O) +
    \frac{A}{\pi(P)(W)}\{Y - \mu(P)(1, W)\}
    - \frac{1-A}{1-\pi(P)(W)}\{Y - \mu(P)(0, W)\}.
  \end{equation*}
\end{theorem}

Theorem~\ref{theorem:anti-target-diff} implies that we can use
standard influence-function based techniques to construct estimators
of $\Gamma$ at points at which $\gamma$ is continuous. For instance,
the one-step estimator
\begin{equation}
  \label{eq:3}
  \hat{\Gamma}_n({\alpha})
  = \Gamma(\hat{P}_n)({\alpha}) +
  \empmeas{[\upsilon_{\alpha}(\hat{P}_n)]},
\end{equation}
will be regular, asymptotically linear and efficient under suitable
regularity conditions
\citep{pfanzagl1985contributions,kennedy2022semiparametric,hines2022demystifying}. Appendix~\ref{sec:proof-theorem-2}
provides results and regularity conditions needed for deriving the
asymptotic distribution of the one-step estimator of $\Gamma$, and we
return to these conditions in Section~\ref{sec:appr-estim}. Here, we
instead define the Grenander estimator derived from the
one-step estimator $\hat{\Gamma}_n(\alpha)$. Following
\cite{westling2020unified}, we define the greatest convex minorant
operator
\( \mathrm{GCM} \colon \mathcal{D}_{[-1,1]} \rightarrow
\mathcal{D}_{[-1,1]} \), where the function \( \mathrm{GCM}(G) \) is
the pointwise supremum over all convex functions $H$ such that
\( H \leq G \). The Grenander estimator of \( \gamma \) is then the
derivative of \( \mathrm{GCM}(\hat{\Gamma}_n) \), where
$\hat{\Gamma}_n$ is the function-valued one-step estimator defined
pointwise in equation~\eqref{eq:3}. We use a cross-fitted one-step
estimator, defined as
\begin{equation}
  \label{eq:cf-one-step}
  \hat{\Gamma}_n^{\bullet} = \frac{1}{K}\sum_{k=1}^{K}
  \hat{\Gamma}_n^{k},
  \quad \text{with} \quad
  \hat{\Gamma}_n^{k}(\alpha) =
  \Gamma(\hat{P}_n^{-k})({\alpha}) +
  \mathbb{P}_n^k{[\upsilon_{\alpha}(\hat{P}_n^{-k})]},
\end{equation}
and consider the cross-fitted Grenander estimator
\begin{equation*}
  \hat{\gamma}_n^{\text{Gr}}(\alpha) = 
  \frac{\partial }{\partial
    \alpha}  \mathrm{GCM}(\hat{\Gamma}_n^{\bullet})(\alpha).
\end{equation*}
We use results and condition from \citep{westling2020unified} and a
von~Mises expansion of $\Gamma(\alpha)$ (stated in
Lemma~\ref{lemma:gateaux} in Appendix~\ref{sec:proof-theorem-2}) to
derive the asymptotic distribution of
\( \hat{\gamma}_n^{\text{Gr}}(\alpha) \) under some conditions.

\begin{theorem}
  \label{theorem:gren-est-rates} Let \( \mathcal{P} \) be a model
  fulfilling Assumption~\ref{assum1}. Assume that the partition of the
  data used to construct the cross-fitted one-step estimator
  \( \hat{\Gamma}_n^{\bullet} \) defined in
  equation~\eqref{eq:cf-one-step} fulfills Assumption~\ref{assum2},
  and that the propensity estimator fulfills Assumption~\ref{assum3}
  for all \( P \in \mathcal{P} \). Assume that
  \( \hat{\tau}_n(w) = \hat{\mu}_n(1, w) - \hat{\mu}_n(0, w) \),
  \( \hat{\eta}_{\alpha, n}(w) = \1{\{\hat{\tau}_n(w) \leq
    \alpha\}}\), and that the nuisance parameter estimators
  \( \hat{\pi}_n \), \( \hat{\mu}_n, \) and \( \hat{\eta}_{a, n} \)
  fulfill the following for all \( P \in \mathcal{P} \).
  \begin{assumptionlist}
  \item\label{assump:gren2}
    \( \|\hat{\mu}_n - \mu(P) \|_P = \smallO_P{(1)} \).
  \item\label{assump:gren1} For each sufficiently small $\delta >0$,
    it holds that \( \mathcal{K}_n(\delta;P) = \smallO_P{(1)} \) where
    \( \mathcal{K}_n(\delta;P) \) is defined in
    equation~\eqref{eq:big-K-term} of
    Appendix~\ref{sec:proof-theor-refth}.
  \end{assumptionlist}
  For any \( P \in \mathcal{P} \) for which
  $\gamma(P)$ is differentiable in a neighborhood around $\alpha$, it
  holds that
  \begin{equation}
    \label{eq:44}
    n^{1/3}( \hat{\gamma}_n^{\text{Gr}}(\alpha) - \gamma(P))
    \rightsquigarrow (2 \sigma_P(\alpha) \gamma'(P)(\alpha))^{2/3} Z,
  \end{equation}
  where
  \begin{equation*}
    Z  = \argmin_{u \in \R}\{ X(u) + u^2\},
  \end{equation*}
  for \( X \) a standard two-sided Brownian motion with \( X(0) =0 \),
  \begin{equation*}
    \sigma_P^2(\alpha)   =
    \E_P{\left[ 
        \frac{\left( Y- \mu(P)(A,W)   \right)^2}{\pi(P)(W)^A (1-\pi(P)(W))^{1-A}}  
        \midd \tau(P)(W) = \alpha
      \right]},
  \end{equation*}
  and $\gamma'(P)(\alpha)$ is the derivative of \( \alpha \mapsto \gamma(P)(\alpha) \).
\end{theorem}

Assumption~\ref{assump:gren1} is taken directly from
\cite{westling2020unified} as we were not able to derive simple
conditions that could be used to easily verify the condition. The
conditions involves bounding an empirical process term and controlling
a remainder term. A crude analysis (not shown here) demonstrates that
this condition would hold under the strong assumption that
\( \| \hat{\tau}_n - \tau(P) \|_P = \smallO_P{(n^{-1/3})} \). Under
this strong assumption, the naive plug-in estimator defined in
Section~\ref{sec:plug-estimator} would achieve a faster rate of
convergence than the Grenander estimator. However, even in this case,
Theorem~\ref{theorem:gren-est-rates} would be of interest as it
establishes a non-trivial asymptotic limit distribution for the
Grenander estimator, which can be used to construct asymptotically
valid confidence intervals. Furthermore, we conjecture that
Assumption~\ref{assump:gren1} can be established under weaker
assumptions, which is also indicated by our numerical experiments in
Section~\ref{sec:numer-exper}.

The distribution of \( Z \) in display~\eqref{eq:44} is referred to as
the standard Chernoff distribution, which is a well understood
distribution. In particular, its quantiles and percentiles have been
tabulated \citep{groeneboom2001computing} and hence, if we can
estimate $\sigma^2(\alpha)$ and $\gamma'(\alpha)$ consistently it is
straightforward to produce confidence intervals using
Theorem~\ref{theorem:gren-est-rates} and Slutsky's lemma. To estimate
$\gamma'(\alpha)$ and $\sigma^2(\alpha)$, we propose to use the
level-one data that is implicitly constructed when employing the
cross-fitted one step estimator. More formally, define the level-one
data
\begin{equation*}
  \left\{
    (\hat{\mu}_i, \hat{\pi}_i, \hat{\tau}_i)  =  (\hat{\mu}_n^{-k}(A_i, W_i), \hat{\pi}_n^{-k}(W_i),
    \hat{\tau}_n^{-k}(W_i)) :
    i \not \in \mathcal{D}_k, k = 1, \dots, K
  \right\},
\end{equation*}
where \( \mathcal{D}_k \) denotes the \( k \)'th fold of the data. We
propose to estimate $\sigma^2$ using a nearest neighbor algorithm
applied to the pseudo-outcome and pseudo-covariates
\begin{equation*}
  \left\{
    \left(
      \frac{\left( Y_i- \hat{\mu}_i   \right)^2}{\hat{\pi}_i^{A_i}
    (1-\hat{\pi}_i)^{1-A_i}} , \,
  \hat{\tau}_i
\right) : i = 1, \dots, n
\right\}.
\end{equation*}
Similar, we propose to estimate $\gamma'(\alpha)$ with a kernel
density estimator applied to
\( \{\hat{\tau}_i : i = 1, \dots, n \} \). The Chernoff distribution
is a symmetric mean-zero distribution with \( 97.5 \% \) percentile
approximately equal to \( q_{0.975} \approx 0.9982 \)
\citep{groeneboom2001computing}. A confidence interval for
$\gamma(\alpha)$ is hence obtained as
\begin{equation*}
  \hat{\gamma}_n^{\text{Gr}}(\alpha) - q_{0.975} 
  \left\{
    \frac{2
      \hat{\sigma}_n(\alpha) \hat{\gamma}_n'(\alpha)}{\sqrt{n}}
  \right\}^{2/3} ,
  \quad
  \hat{\gamma}_n^{\text{Gr}}(\alpha) + q_{0.975} 
  \left\{
    \frac{2
      \hat{\sigma}_n(\alpha) \hat{\gamma}_n'(\alpha)}{\sqrt{n}}
  \right\}^{2/3} ,
\end{equation*}
where \( \hat{\sigma}_n(\alpha)  \) and \( \hat{\gamma}_n'(\alpha) \)
are the kernel-based estimators fitted on the level-one data introduced above.

\subsection{Spline approximation estimators}
\label{sec:appr-estim}

By Theorem~\ref{theorem:target-nondiff}, estimation of the parameter
$\gamma$ is challenging. One strategy for estimation is to instead
settle for estimation of an approximation of
$\gamma$. % This is similar
% in spirit to the approach taken by, e.g., \cite{semenova2021debiased}
% and \cite{chernozhukovFisher}.
%%
% Also add other reference here, should be something about projections
% or similar in vdLaan+Rose 
For instance, we might consider estimating the best \( r \)-order
spline approximation of $\gamma$. Here we consider first order splines
with \( L \) predetermined inner knot-points, i.e., we consider
estimating the best piece-wise linear approximation of $\gamma$ over
\( L+1 \) predetermined intervals.  We allow for the possibility of
estimating the best linear approximation of $\gamma$ only over an
interval \( [a_0, a_{L+1}] \subseteq [-1, 1] \), and thus consider the
best linear approximation determined by the knot-points
\( -1 \leq a_0 < a_1 < \dots < a_{L+1} \leq 1 \). A basis for the
space of first order splines with these knot-points are the hat
functions
\begin{equation*}
  H_l(u) =
  \frac{u - a_{l-1}}{a_l - a_{l-1}} \1_{(a_{l-1}, a_l]}(u)
  +
  \frac{a_{l+1}-u}{a_{l+1} - a_{l}} \1_{(a_{l}, a_{l+1}]}(u),
\end{equation*}
for \( l=1, \dots, L \), and
\begin{equation*}
  H_0(u) =
  \frac{a_{1}-u}{a_{1} - a_{0}} \1_{(a_{0}, a_{1}]}(u),
  \quad \text{and} \quad
  H_{L+1}(u) =    \frac{u - a_{L}}{a_{L+1} - a_{L}} \1_{(a_{L}, a_{L+1}]}(u).
\end{equation*}
The best first order spline approximation of a function
\( \gamma \in \mathcal{D}_{[-1,1]} \) with respect to the
\( \L{2}{m} \)-norm is the function
\begin{equation}
  \label{eq:spline-estimand}
  \gamma^{\#}(\alpha) = \sum_{l=0}^{L+1} H_l(\alpha)
  \gamma^{\#}_l ,
\end{equation}
where the coefficient vector
\( (\gamma^{\#}_0, \dots, \gamma^{\#}_{L+1})^{\top} \in
\R^{L+2}\) is equal to
\begin{equation*}
  M^{-1} \zeta,
  \quad \text{where} \quad 
  M_{lj} = \langle H_l, H_j \rangle
  \quad \text{and} \quad
  \zeta_l = \langle H_l, \gamma \rangle,
  \quad \text{for } l,j = 0, \dots,  L+1,
\end{equation*}
with \( \langle \blank, \blank \rangle \) the \( \L{2}{m} \)-inner
product where \( m \) denotes Lebesgue measure
\citep{de1978practical}. Some examples of best linear approximations
$\gamma^{\#}$ of different sublevel functions $\gamma$ are shown in
Figure~\ref{fig:illu-gamma-spline}.

\begin{figure}
  \centerline{\includegraphics[width=1\linewidth]{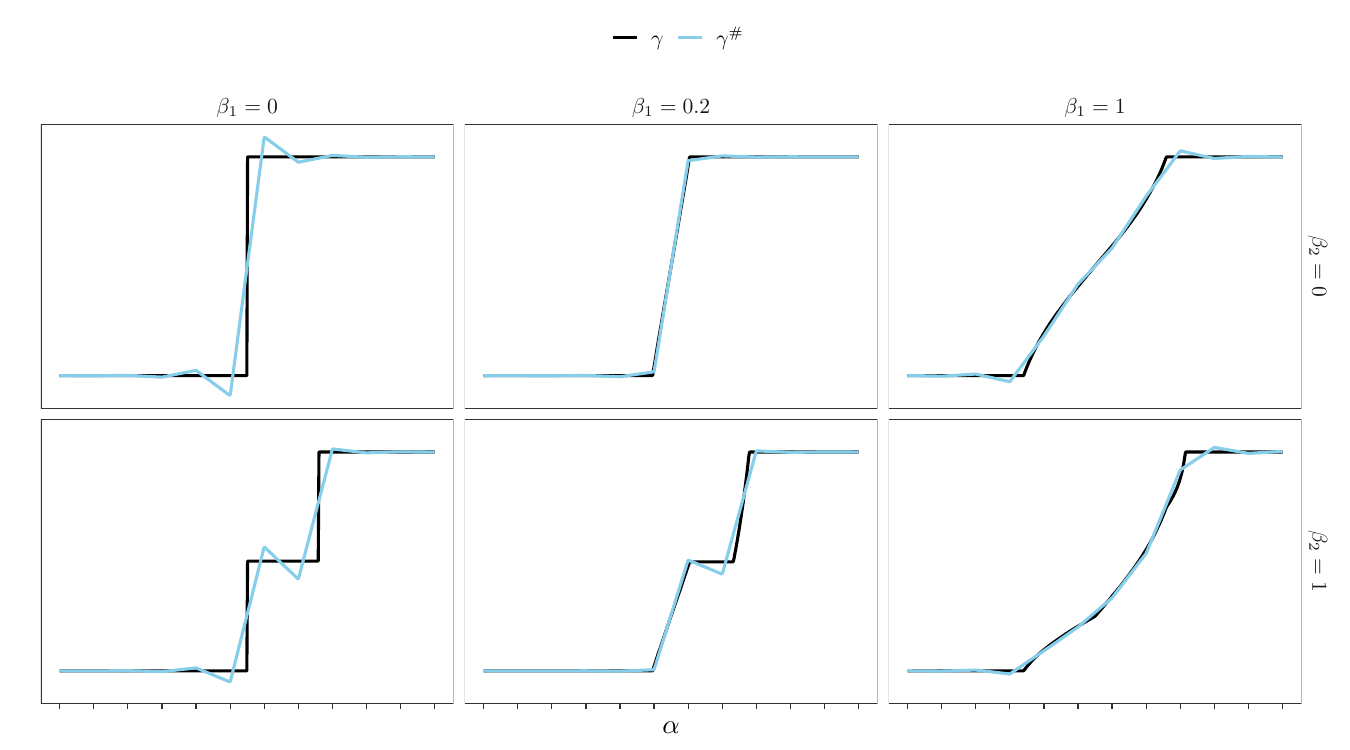}}
  \caption[]{The blue line show the best piece-wise linear
    approximations to the black curve ($\gamma$) for different
    data-generating mechanisms defined in
    equation~\eqref{eq:simple-dgm} in
    Section~\ref{sec:notation-estimand} with \( p=0.5 \). The ticks at
    the \( x \)-axis are the knot points used for the linear
    approximation.}
  \label{fig:illu-gamma-spline}
\end{figure}

The above calculations demonstrate that the function-valued parameter
\( \gamma^{\#} \) can be estimated by estimating the vector
$\zeta \in \R^{L+2}$.  We can use integration by parts to show that
\begin{equation}
  \label{eq:30}
  \begin{split}
    \zeta_0
    &= -\Gamma(a_0) +\frac{1}{a_{1}-a_0} \int_{a_0}^{a_{1}} \Gamma(u)
      \diff u,
    \\
    \zeta_l
    & =  \frac{-1}{a_{l}-a_{l-1}} \int_{a_{l-1}}^{a_{l}} \Gamma(u)  \diff
      u + \frac{1}{a_{l+1}-a_l} \int_{a_l}^{a_{l+1}} \Gamma(u)  \diff u,
      \quad \text{for} \quad   l=1, \dots, L,
    \\
    \zeta_{L+1}
    &= \Gamma(a_{L+1})   -  \frac{1}{a_{L+1}-a_{L}} \int_{a_{L}}^{a_{L+1}} \Gamma(u)  \diff u.
  \end{split}
\end{equation}
To estimate \( \zeta \), we use the cross-fitted one-step estimator of
\( \Gamma \) defined in equation~\eqref{eq:cf-one-step} and plug it
into the expressions in equation~\eqref{eq:30}. We show that, under
some rate conditions on the nuisance parameter estimators, this
provides an efficient, regular asymptotically linear estimator as long
as the knot-points \( a_1, \dots, a_L \) are placed at points at which
\( \gamma \) is continuous. In the following, we abuse notation
slightly and write \( \zeta \) as a function of \( \Gamma \).  We use
\( \| \blank \|_P\) to denote the \( \LP \)-norm.

\begin{theorem}
  \label{theorem:spline-coef}
  Let \( \mathcal{P} \) be a model with saturated tangent space and
  fulfilling Assumption~\ref{assum1}, and let
  \( -1\leq a_0 < a_1 \dots < a_{L+1} \leq 1 \) be fixed
  knot-points. Assume that the partition of the data used to construct
  the cross-fitted one-step estimator \( \hat{\Gamma}_n^{\bullet} \)
  defined in equation~\eqref{eq:cf-one-step} fulfills
  Assumption~\ref{assum2}, and that the propensity estimator fulfills
  Assumption~\ref{assum3} for all \( P \in \mathcal{P} \). Assume that
  \( \hat{\tau}_n(w) = \hat{\mu}_n(1, w) - \hat{\mu}_n(0, w) \),
  \( \hat{\eta}_{\alpha, n}(w) = \1{\{\hat{\tau}_n(w) \leq
    \alpha\}}\), and that the nuisance parameter estimators
  \( \hat{\pi}_n \), \( \hat{\mu}_n, \) and
  \( \hat{\eta}_{\alpha, n} \) fulfill the following for all
  \( P \in \mathcal{P} \).
  \begin{assumptionlist}
  \item\label{assump:spline1}
    \( \|\hat{\pi}_n - \pi(P) \|_P = \smallO_P{(1)} \), and
    \( \|\hat{\eta}_{a, n} - \eta_{a}(P) \|_P = \smallO_P{(1)} \) for
    all \( a \in \{a_0, \dots, a_{L+1}\} \).
  \item\label{assump:spline2}
    \( \| \hat{\pi}_{n} - \pi(P) \|_P \| \hat{\mu}_{n} - \mu(P) \|_P =
    \smallO_P{(n^{-1/2})}\) and
    \( \| \hat{\mu}_{n} - \mu(P) \|_P^2 = \smallO_P{(n^{-1/2})}\).
  \item\label{assump:spline3}
    \( P{[ |\hat{\eta}_{a, n} - \eta_{a}(P)| |\hat{\tau}_{n} - \tau(P)|]} = \smallO_P{(n^{-1/2})}\) for \( a \in \{a_0, a_{L+1}\} \).
  \end{assumptionlist}
  For any \( P \in \mathcal{P} \) for which
  $\alpha \mapsto \gamma(P)(\alpha)$ is continuous at each of the
  knot-points \( a_0, \dots , a_{L+1} \), the estimator
  \( \zeta(\hat{\Gamma}_n^{\bullet}) \) is locally efficient and
  regular and satisfies the asymptotic expansion
  \begin{equation}
    \label{eq:45}
    \zeta(\hat{\Gamma}_n^{\bullet}) - \zeta = \empmeas{[\dot{\zeta}]} + \smallO_P{(n^{-1/2})},
  \end{equation}
  where $\dot{\zeta}$ is the (vector-valued) efficient influence function of
  \( \zeta \),
  \begin{align*}
    \dot{\zeta}_0(O)
    &= -\upsilon_{a_0}(O) +\frac{1}{a_{1}-a_0} \int_{a_0}^{a_{1}} \upsilon_{u}(O)
      \diff u,
    \\
    \dot{\zeta}_l(O)
    & =  \frac{-1}{a_{l}-a_{1-1}} \int_{a_{l-1}}^{a_{l}} \upsilon_u(O)  \diff
      u + \frac{1}{a_{l+1}-a_l} \int_{a_l}^{a_{l+1}} \upsilon_u(O)  \diff u,
      \quad \text{for} \quad   l=1, \dots, L,
    \\
    \dot{\zeta}_{L+1}(O)
    &= \upsilon_{a_{L+1}}(O)   -  \frac{1}{a_{L+1}-a_{L}}
      \int_{a_{L}}^{a_{L+1}} \upsilon_u(O)  \diff u,
  \end{align*}
  where $\upsilon_{\alpha}$ was defined in
  equation~\eqref{eq:anti-target-can-grad}.
\end{theorem}

Theorem~\ref{theorem:spline-coef} imposes conditions on the nuisance
parameter estimators $\hat{\mu}_n$, $\hat{\pi}_n$, and
$\hat{\eta}_n$. Besides consistency, certain rate of convergence
requirements (conditions~\ref{assump:spline2}
and~\ref{assump:spline3}) are imposed. Importantly, these conditions
involve product terms and squared norms, which means that the
conditions are weak enough to allow for estimators of $\mu$ and $\pi$
that converge at rates slower than the parametric rate. It is worth
noting that condition~\ref{assump:spline3} is only required at the two
boundary knot-points and is thus trivially satisfied if we consider
the best first-order spline approximation across the whole interval
\( [-1,1] \), because in this case
$\hat{\eta}_{-1, n} = \eta_{-1}(P) = 0$ and
$\hat{\eta}_{1, n} = \eta_{1}(P) = 1$. More generally, Hölder's
inequality and a local smoothness assumption on $\gamma$ can be used
to give a crude bound in terms of the supremum norm of
\( \tau - \hat{\tau}_n \).

\begin{lemma}
  \label{lemma:rate-local-sup-norm}
  If $\gamma(P)$ is differentiable in a neighborhood around
  \( \alpha \) then  \( P{[ |\hat{\eta}_{\alpha, n} - \eta_{\alpha}| |\hat{\tau}_{n} -
    \tau|]} = \bigO{(\| \hat{\tau}_{n} - \tau \|_{\infty}^2)} \) 
\end{lemma}

Theorem~\ref{theorem:spline-coef} suggests the estimator
\begin{equation}
  \label{eq:spline-approx-est}
  \hat{\gamma}^{\#}_n(\alpha) = \sum_{l=0}^{L+1} H_l(\alpha)
  \hat{\gamma}_{l,n}^{\#},
  \quad \text{where} \quad
  \hat{\gamma}_{l,n}^{\#} = 
  \left(
    M^{-1} \zeta(\hat{\Gamma}_n^{\bullet})
  \right)_l.
\end{equation}
For at fixed point $\alpha$, the estimator
\( \hat{\gamma}^{\#}_n(\alpha) \) is a simple linear transformation of
the vector \( \zeta(\hat{\Gamma}_n^{\bullet}) \), and hence the
asymptotic expansion in equation~\eqref{eq:45} immediately allows us
to use the delta method to provide pointwise confidence intervals. The
asymptotic expansion can also be used to construct a simultaneously
valid confidence band with standard methods from density and
regression estimation \citep[e.g.,][section
5.7]{wasserman2006all}. Details on these confidence band are given in
Appendix~\ref{sec:conf-bands-splin}.

The function \( \gamma^{\#} \) is not guaranteed to be neither
monotone nor contained within \( [0,1] \), and can hence in some cases
provide an implausible proxy for a cumulative distribution
function. This is seen in some of the panels in
Figure~\ref{fig:illu-gamma-spline}. We can address this by considering
the best spline approximation of $\gamma$, under the constraints that
the spline should be monotone and bounded by \( 0 \) and \( 1 \). This
constrained spline is defined by the coefficient vector solving the
convex quadratic optimization problem
\begin{equation*}
  \min_{x \in [0,1]^{L+2}} x^{\top} G x - 2 \zeta^{\top} x
  \quad \text{such that} \quad
  x_{l}  \leq   x_{l+1}, \; l = 0, \dots, L.
\end{equation*}
Hence, an estimate of $\zeta$ also provides an estimate of the
constrained spline through the solution of this convex problem. The
asymptotic expansion from Theorem~\ref{theorem:spline-coef} could also
be use to obtain confidence intervals, but the constrained nature of
the problem makes this more involved
\citep{geyer1994asymptotics,shapiro2000asymptotics}. We consider here
a different approach which is based on monotonizing the spline
estimate from equation~\eqref{eq:spline-approx-est} and its associated
confidence bands directly using the rearrangement approach suggested
by \cite{chernozhukov2009improving}. In practice, this amounts so
reporting the sorted values of the \( \hat{\gamma}^{\#} \) evaluated
at a dense grid of points. We refer to this as a monotonized spline
estimator, which is formally defined as the function
\begin{equation}
  \label{eq:spline-mono}
  \alpha \longmapsto \inf_{x \in \R} {
    \left\{
      \int_{-1}^{1} \1{\{\hat{\gamma}^{\#}(u) \leq x\}} \diff u \geq \alpha
    \right\}}.
\end{equation}
The associated monotonized confidence bands are defined similarly.

\section{Numerical experiments}
\label{sec:numer-exper}

\begin{figure}[h]
  \centerline{\includegraphics[width=1\linewidth]{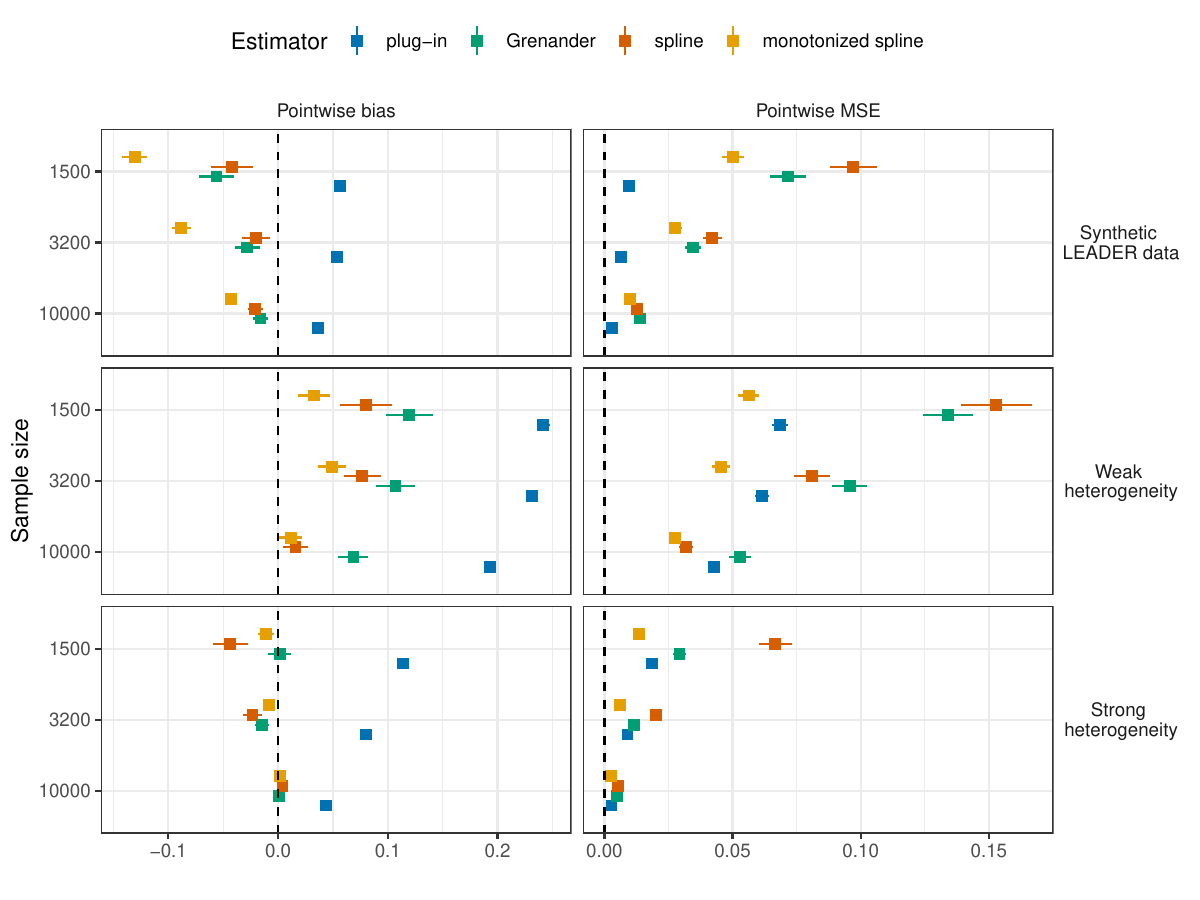}}
  \caption[]{The pointwise bias and mean squared error (MSE) of the
    four estimators of $\gamma$ defined in
    Section~\ref{sec:estim-thro-antid} for varying sample size and
    degree of heterogeneity evaluated at the value $\alpha=0.01$. The
    results are averages across 1000 simulated data sets and the error
    bars are used to quantify the Monte Carlo uncertainty.}
  \label{fig:bias-mse}
\end{figure}

\begin{figure}[h]
  \centerline{\includegraphics[width=1\linewidth]{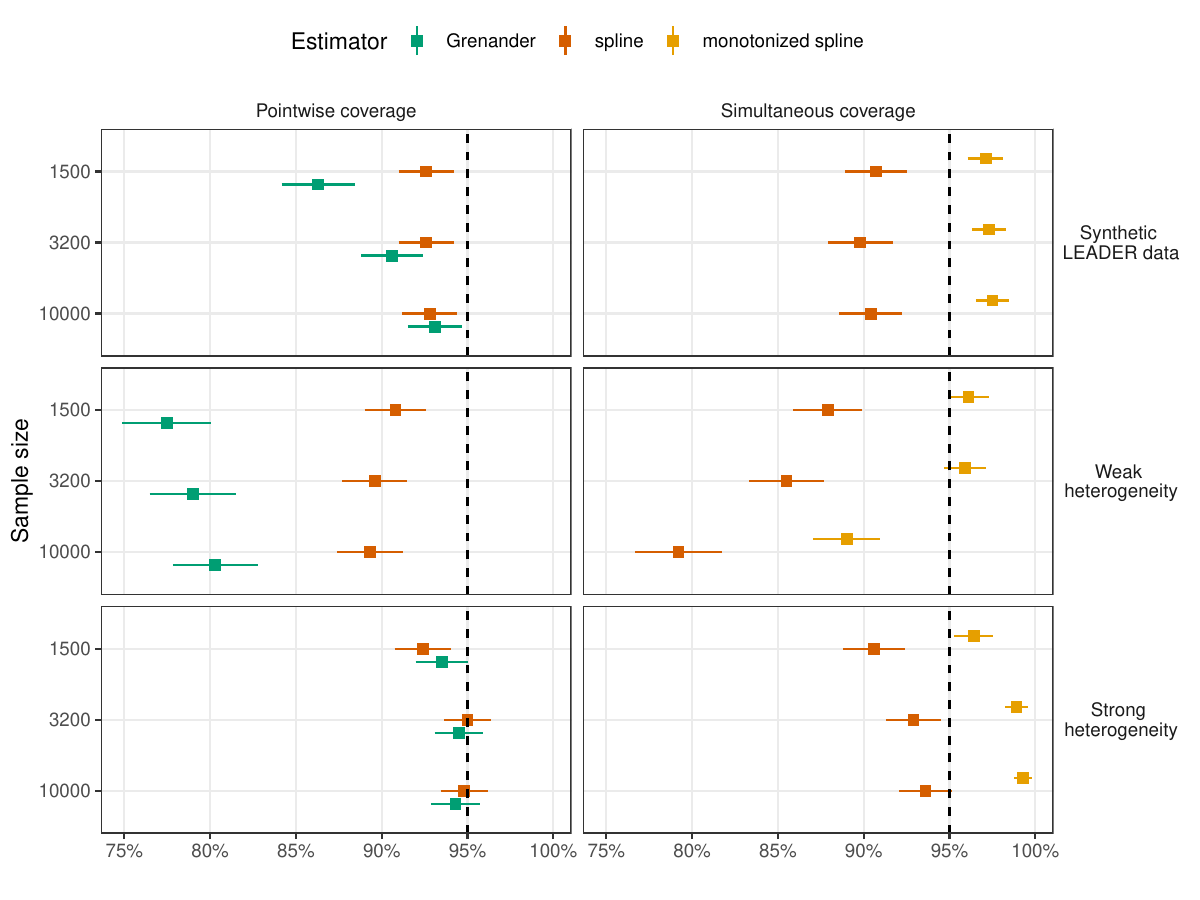}}
  \caption[]{The left panel shows the pointwise coverage for the
    pointwise confidence intervals provided by the Grenander and the
    spline estimator evaluated at the $\alpha=0.01$. The right panel
    shows the simultaneous coverage of the confidence bands for the
    raw and monotonized spline estimators. The results are averages
    across 1000 simulated data sets and the error bars are used to
    quantify the Monte Carlo uncertainty.}
  \label{fig:coeverage}
\end{figure}

To investigate the performance of the estimators introduced in
Section~\ref{sec:estim-thro-antid}, we conducted a numerical
experiment with a data-generating mechanism based on data from the
LEADER (Liraglutide Effect and Action in Diabetes: Evaluation of
Cardiovascular Outcome Results) trial \citep{marso2016liraglutide}. We
simulated five binary and three continuous baseline covariates, which
were generated sequentially using logistic regression and linear
normal models with main effects only, where the parameters used were
estimated from the trial data. Treatment was simulated
independently. The model for the outcome was constructed based on a
logistic regression with main effects and all second order treatment
interactions, with parameter values estimated from the trial data. We
conducted three numerical experiments by simulating data with this
setup, where we varied the parameters of the outcome model as follows:
In the first experiment (LEADER synthetic data) we used the parameter
values estimated from the trial without alterations; in the second
experiment (weak heterogeneity) we attenuated the main effect of
treatment and removed all interactions; and in the third experiment
(strong heterogeneity) we amplified both the main effect of treatment
and the interactions. The three sublevel functions $\gamma$ resulting
from these data generating mechanism are shown in
Figure~\ref{fig:sim-true-curves} in
Appendix~\ref{sec:addit-inform-about}.

In all experiments, the propensity model $\pi$ was estimated by the
marginal probability of receiving treatment. The outcome model $\mu$
was estimated using an elastic net \citep{zou2005regularization,friedman2010glmnet}
including all second- and third-order interactions between all
variables, with penalty parameter selected using cross-validation. The
CATE function $\tau$ and the sublevel indicator function
$\eta_{\alpha}$ were both estimated with plug-in estimators based on
the elastic net estimate of $\mu$. Based on these nuisance parameter
estimates, we calculated the following four estimators of $\gamma$:
The plug-in estimator (plug-in) defined
Section~\ref{sec:plug-estimator}; the Grenander estimator (Grenander)
defined in Section~\ref{sec:gren-type-estim}; the spline estimator
(spline) defined in equation~\eqref{eq:spline-approx-est} in
Section~\ref{sec:appr-estim}; and the monotonized spline estimator
(monotonized spline) defined in equation~\eqref{eq:spline-mono} in
Section~\ref{sec:appr-estim}. We estimated the sublevel function
$\gamma$ restricted to the interval \( [-0.05,0.1] \). For the spline
estimators we used an equidistant grid in the interval with 10 knot
points. We considered sample sizes of \( n \in \{1500, 3200,10000\} \)
and repeated each simulation \( 1000 \) times.

We evaluated pointwise performance of all four estimators in terms of
bias and mean squared error (MSE). For the Grenander estimator and the
raw (non-monotonized) spline estimator we also evaluated the coverage
of the pointwise confidence intervals provided by the asymptotic
distribution results provided by Theorems~\ref{theorem:gren-est-rates}
and~\ref{theorem:spline-coef}, respectively. The point of evaluation
was chosen to be $\alpha=0.01$, which was close to the estimated
average treatment effect in the original trial data. For the spline
estimator, we also evaluated the simultaneous coverage performance the
confidence band described in Appendix~\ref{sec:conf-bands-splin} and
its monotonized version introduced at the end of
Section~\ref{sec:appr-estim}.

The left panel of Figure~\ref{fig:bias-mse} demonstrates that the
Grenander estimator successfully removes bias when compared to the
plug-in estimator in all experiments. Both spline estimators similarly
demonstrate a clear debiasing effect. For the spline estimators, the
bias also depends on how close the best piece-wise linear
approximation is to the true sublevel function; in all our
experiments, the approximation is quite good for the point of
evaluation, c.f., Figure~\ref{fig:sim-true-curves} in
Appendix~\ref{sec:addit-inform-about}. Whether monotonizing the spline
improves bias depends on the data-generating distribution. While the
Grenander and the spline estimators improve bias, this comes at the
cost of an increase in variance as compared to the plug-in estimator,
which is evident from the MSE shown in the right panel of
Figure~\ref{fig:bias-mse}. The Grenander and raw spline estimator
exhibit larger MSE than the plug-in estimator in most settings and for
most sample sizes. As expected, monotonizing the spline improves
variance, and the monotonized spline have smaller MSE than the plug-in
estimator for both the weak and the strong heterogeneity settings.

Figure~\ref{fig:coeverage} demonstrates that the coverage of the
pointwise confidence intervals for the Grenander and the raw spline
estimator are fairly close to the nominal level for large sample sizes
and settings with non-weak interaction. In the setting with only a
weak interaction, the coverage is anti-conservative and far from the
nominal level even for large sample sizes. The simultaneous coverage
for the raw and the monotonized spline based confidence bands are in
most cases anti-conservative and convervative, respectively, except
for the weak interaction setting where performance deteriorates with
sample size and coverage becomes increasingly anti-conservative for
both confidence bands. Neither types of simultaneous confidence bands
are achieving exact nominal coverage, which is to be expected because
coverage is measured with respect to the true $\gamma$ function while
the splines estimators are targeting an approximation of $\gamma$.

\section{Application to the LEADER trial}
\label{sec:application}

\begin{figure}[h]
  \centerline{\includegraphics[width=1\linewidth]{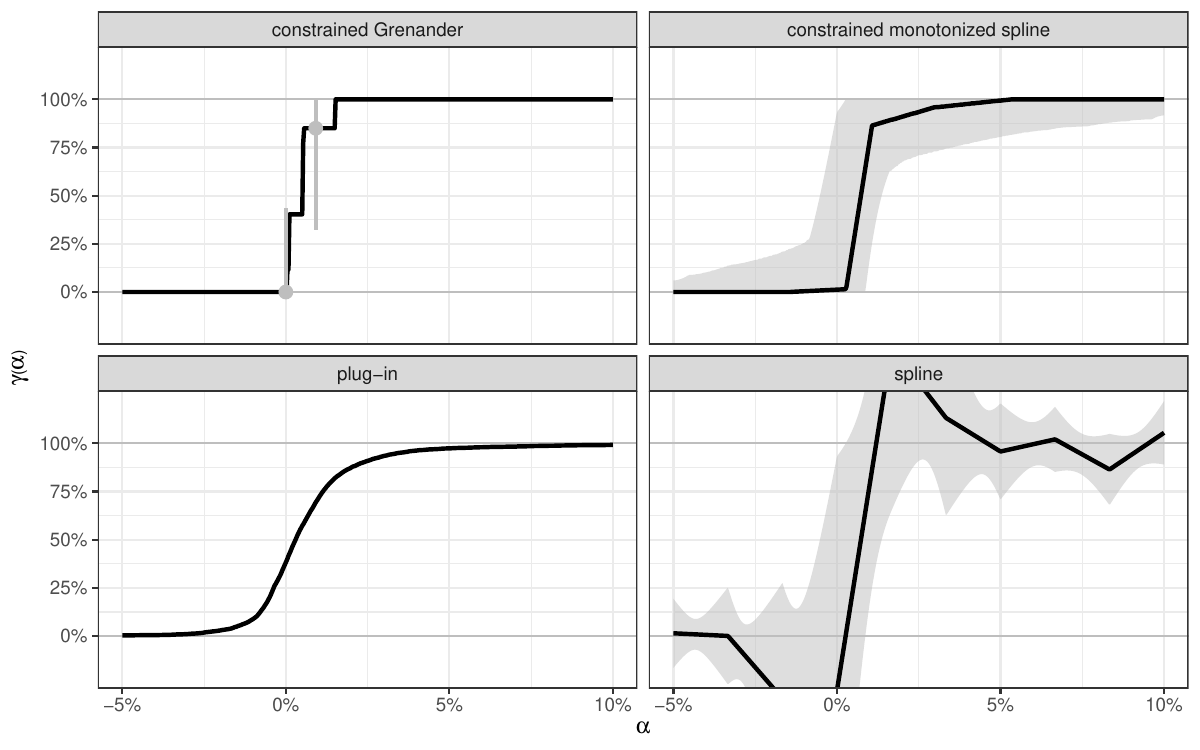}}
  \caption[]{Four different estimates of the sublevel function
    $\gamma$ applied to the LEADER data set. For the spline and the
    constrained monotonized spline estimators, the gray ribbons
    denotes simultaneously valid 95\%-confidence band. For the
    constrained Grenander estimator, the gray point ranges indicate
    pointwise valid 95\%-confidence intervals at $\alpha=0$ and
    \( \alpha=0.01 \).}
  \label{fig:leader-appl}
\end{figure}

To illustrate our method, we applied the suggested estimators to data
from a 1:1 randomized trial which ran from 2010 to 2015
\citep{marso2016liraglutide}. The trial investigated the
cardiovascular effect of liraglutide as compared to placebo. For
simplicity, we used a composite event consisting of nonfatal
myocardial infarction, nonfatal stroke, or death from any cause, and
we investigate the effects of being randomized to treatment compared
to being randomized to placebo, corresponding to an intention-to-treat
estimand. We defined the binary outcome to be whether the composite
event had occurred within 3 years of randomization. The study included
8,563 patients out of which 26 patient where censored before 3
years. For the sole purpose of illustration, these 26 censored
patients were removed from the analysis. As baseline information the
following variables were used: sex (male or female), age group (40-60,
60-80, or 80-90 years), diabetes duration (more or less than 11
years), current insulin use (yes or no), heart failure (yes or no),
estimated glomerular filtration rate (eGFR), body mass index (above or
below \( 30\text{kg}/\text{m}^2 \)), and HbA1c.

To estimate the outcome model $\mu$ defined in equation~\eqref{eq:11},
we used a super learner \citep{van2007super, breiman1996stacked} based
on a library consisting of a logistic regression with main effects, a
random forest \citep{breiman2001randomforests,wright2017ranger}, an
elastic net including all second- and third-order interactions between
all variables \citep{zou2005regularization,friedman2010glmnet}, and a
gradient-boosted tree algorithm \citep{chen2016xgboost}. We used 5
fold cross-validation to construct the super learner. Because
treatment was randomized at baseline, the propensity score was
estimated as the marginal probability of treatment. As estimators of
the sublevel function $\gamma$, we considered the plug-in estimator,
the spline estimator, and constrained versions of both the Grenander
and the monotonized spline estimator. The constrained versions
truncate the estimates and confidence intervals pointwise to be
contain in the interval \( [0,1] \).  We used 3 folds for
cross-fitting.

The results are shown in Figure~\ref{fig:leader-appl}. We see that all
estimators suggest that the expected treatment effects are centered
fairly closely around the average treatment effect (estimated to be
around \( 1\% \)), with the plug-in estimator suggesting the largest
amount of heterogeneity. The proportion of the trial population for
which no treatment benefit is expected is the number $\gamma(0)$, and
the spline, the constrained monotonized spline, and the constrained
Grenander estimator all estimate this close to zero. The pointwise
95\%-confidence interval for this number provided by the constrained
Grenander estimator is \( [0\%, 44\%] \). The pointwise confidence
interval for the value \( 1-\gamma(0.01) \), which approximately
corresponds to the number of patients with an expected treatment
effect above the average treatment effect, is \( [0\%, 68\%] \). The
confidence bands provided by the spline estimator and its constrained
monotonized version are compatible with an abruptly changing $\gamma$,
which indicates that most patient have a treatment effect close the
the average effect.

%% Should write something about being critical about these values
%% because the estimators can struggle in cases with low degree of
%% heterogeneity.

\section{Discussion}
\label{sec:discussion}

We have argued that the sublevel function introduced in this paper
is an interesting function-valued target parameter that has a clear
interpretation and encodes information about the type and degree of
heterogeneity in a population in a visually simple manner. We have
shown that estimating this parameter is challenging, in the sense that
the parameter of interest is not pathwise differentiable. A remaining
question is to quantify exactly how difficult the problem is in terms
of minimax optimal rates of convergence. In models where the sublevel
function $\gamma$ is differentiable, we conjecture that the
minimax rate of convergence is of the order \( n^{-1/3} \), while we
expect the problem to be as difficult as estimating the CATE function
itself in models where $\gamma$ is allowed to be discontinuous. The
works of \cite{kennedy2023towards} and \cite{bonvini2023minimax} would
be the starting points for answering this question.

We have suggested three different estimation strategies, two of which
were based on debiased machine learning techniques and provided a way
to construct confidence intervals. While we found these methods to
successfully decrease bias as compared to a plug-in estimator, this
gain was at the cost of an increase in variance so large that the MSE
was still in general the lowest for the plug-in estimator. We also
found that a large sample size and some degree of heterogeneity were
needed for our estimators to perform well. It is of interest to
investigate other estimation strategies, in particular, strategies
that would also perform well under the null hypothesis of no
heterogeneity, where the sublevel function is discontinues. Another
possible estimation strategy, that we have not included in this paper,
is to use the ideas of orthogonal statistical learning
\citep{foster2019orthogonal} and employ isotonic regression to a
suitable orthogonalized pseudo-outcome.

Another line of attack for the construction of reliable estimators is
to consider more sophisticated nuisance parameter
estimators. \cite{kennedy2023towards} showed that CATE function
estimation can be easier than estimation of the outcome regression
itself, and hence our suggested sublevel function estimators might
improve with the use of dedicated CATE learners. Furthermore, it has
been argued that classification is easier than regression
\citep[e.g.,][section~6.7]{devroye1996probabilistic}, and thus one
might consider alternative strategies for estimating the sublevel-set
indicator function \( \eta_{\alpha} \).

\section{Funding}
\label{sec:funding}

Funded by the European Union under contract number 101095556. Views
and opinions expressed are, however, those of the authors only and
do not necessarily reflect those of the European Union or European
Health and Digital Executive Agency (HADEA). Neither the European
Union nor the granting authority can be held responsible for
them. This work has received funding from the UK Research and
Innovation.

\section{Acknowledgements}
\label{sec:acknowledgements}

We would like to thank Novo Nordisk for allowing us to analyze the
LEADER data.

\appendix

\section{Some general pathwise differentiability results}
\label{sec:gener-result-pathw}

For any $\sigma$-finite measure $\nu$ on the sample space
\( \mathcal{O} = \mathcal{W} \times \{0,1\} \times \{0,1\} \) and
\( p \in [1, \infty) \), we use \( \L{p}{\nu} \) to denote the
collection of all measurable functions
\( f \colon \mathcal{O} \rightarrow \R \) such that
\( \nu{[f^p]} < \infty \). We use \( \| \blank \|_{\nu, p} \) to
denote the norm
\begin{equation*}
  \| f \|_{\nu, p} = 
  \left(
    \nu{[|f|^p]}
  \right)^{1/p},
\end{equation*}
and we use \( \| \blank \|_{\infty} \) to denote the supremum norm. We
let \( \mathcal{P} \) denote a collection of probability measures on
\( \mathcal{O} \), referred to as a model, which we assume is
dominated by some fixed $\sigma$-finite measure $\nu$. The notion of
pathwise differentiability of a functional
\( \Psi \colon \mathcal{P} \rightarrow \R \) is defined using
quadratic mean differentiable paths
\( \{P_{\epsilon} : \epsilon > 0 \} \subset \mathcal{P} \). A path
\( \{P_{\epsilon}\} \subset \mathcal{P} \) from
\( P \in \mathcal{P} \) is the said to be differentiable in quadratic
mean if there exists an element \( \dot{\ell} \in \LP \) such
\begin{equation*}
  \left\| 
    \frac{\sqrt{p_{\epsilon}} - \sqrt{p} }{\epsilon} -
    \frac{1}{2} \dot{\ell} \sqrt{p}
  \right\|_{\nu, 2}
  \longrightarrow 0,
\end{equation*}
where \( p \) and \( p_{\epsilon} \) are the $\nu$-densities of
\( P \) and \( P_{\epsilon} \), respectively. The function
\( \dot{\ell} \in \LP \) is called the score function of the path
\( \{P_{\epsilon}\} \). A collection \( \{\dot{\ell}\} \) of score
functions is called a tangent set, and the closed linear span of a
tangent set is called a tangent space. A tangent space is always a
subset of the space of all \( P \)-zero mean functions in \( \LP \),
and when a tangent space is equal to this space, we say that the
tangent space is saturated. We formalize the notion of a ``fully
nonparametric'' model \( \mathcal{P} \) as a model for which we can
construct a saturated tangent space using paths contained in
\( \mathcal{P} \). We refer to this simply as a model with a saturated
tangent space.

The statistical parameter
$\Psi \colon \mathcal{P} \rightarrow \R$ is pathwise differentiable at
\( P \in \mathcal{P} \) relative to a tangent set
\( \dot{\mathcal{P}}_P \) if there exists an element $\psi(P) \in \LP$
such that for all \( \dot{\ell} \in \dot{\mathcal{P}}_P \) and a path
\( \{P_{\epsilon}\} \) with score function \( \dot{\ell} \), it holds
that
\begin{equation}
  \label{eq:pathwise-diff-def}
  \frac{\partial }{\partial \epsilon} \Big|_{\epsilon = 0}\Psi(P_{\epsilon}) =
  P{[\psi(P) \dot{\ell}]}.
\end{equation}
Any function $\psi(P) \in \LP$ fulfilling
equation~\eqref{eq:pathwise-diff-def} is called a gradient for $\Psi$
at \( P \). The maximal tangent set is the collection of score
functions of all quadratic mean differentiable paths.  A parameter
will often not be pathwise differentiable relative to the maximal
tangent set, but only relative to a smaller tangent set. Here we
restrict attention to the collection of quadratic mean differentiable
paths fulfilling the regularity conditions stated below in
Assumption~\ref{assum-path}. Any density \( p \) (with respect to the
dominating measure \( \nu \)) for the distribution
\( P \in \mathcal{P} \) factorizes as
\begin{equation*}
  p(y,a,w) = \mu(P)(a,w)^y(1-\mu(P)(a,w)))^{y-1} \pi(P)(w)^a(1-\pi(P)(w))^{1-a} \rho(P)(w),
\end{equation*}
where $\mu$ and $\pi$ were defined in equations~\eqref{eq:10}
and~\eqref{eq:11}, and $\rho$ denotes the density of the marginal
distribution of \( W \). For a path
\( \{P_{\epsilon} : \epsilon > 0\} \) we use the shorthand notation
$\mu_{\epsilon} = \mu{(P_{\epsilon})}$,
$\pi_{\epsilon} = \pi{(P_{\epsilon})}$, and
$\rho_{\epsilon} = \rho{(P_{\epsilon})}$.
\begin{assumption}
  \label{assum-path}
  The exists a constant \( C \leq \infty \) such that for all
  elements of the path \( \{P_{\epsilon} : \epsilon >0\} \)
  from \( P \), the following holds.
  \begin{enumerate}[label=(\roman*), topsep=0pt]
  \item\label{item:3} The Radon-Nikodym derivative
    \( {\rho{(P_{\epsilon})}}/{\rho{(P)}} \) is uniformly bounded by
    \( C \) for all $\epsilon > 0$.
  \item\label{item:4} The partial derivatives of
    \( \mu_{\epsilon}(a, w) \), \( \pi_{\epsilon}(w) \), and
    $\rho_{\epsilon}(w)$ with respect to $\epsilon$ exist at
    all $\epsilon \geq 0$ for all $a \in \{0,1\}$ and
    \( w \in \mathcal{W} \) and are uniformly bounded by
    \( C \).
  \end{enumerate}
\end{assumption}
In the main article, pathwise differentiability is tacitly understood
to mean pathwise differentiability relative to the tangent set derived
from this specific subset of paths. This convention follows
\cite{van1991differentiable,van2000asymptotic,van2002semiparametric}.
The collection of paths fulfilling Assumption~\ref{assum-path}
include, for instance, all paths of the form
\( P_{\epsilon} = (1 + \epsilon h) \cdot P \), where
\( h \in \mathcal{L}_P^2 \) is an arbitrary uniformly bounded function
with \( P{[h]} = 0 \).  As the collection of all uniformly bounded
functions are dense in \( \mathcal{L}_P^2 \), it follows that the
collection of paths fulfilling Assumption~\ref{assum-path} provides a
the basis for a saturated tangent space. Hence, imposing
Assumption~\ref{assum-path} does not restrict a fully nonparametric
model as defined above.

We shall make use of the following result which is adapted from
Lemma~2 in \citep{kennedy2023semiparametric}.
\begin{proposition}
  \label{prop:expansion-gradient}
  Let $\Psi \colon \mathcal{P} \rightarrow \R$ be a
  statistical parameter and \( P \in \mathcal{P} \). Assume
  that $\Psi$ satisfies the expansion
  \begin{equation*}
    \Psi(P) - \Psi(P_{\epsilon}) = (P-P_{\epsilon})[ \psi(P)]
    + R_2(P, P_{\epsilon}),
  \end{equation*}
  for any mean differentiable path
  \( \{P_{\epsilon} : \epsilon > 0 \} \) from
  \( P \in \mathcal{P} \). If
  \( \| \psi(P) \|_{\infty} < \infty \) and
  \begin{equation*}
    \frac{\partial }{\partial \epsilon} \Big|_{\epsilon=0} R_2(P, P_{\epsilon}) = 0,
  \end{equation*}
  for all paths, then $\Psi$ is pathwise differentiable at
  \( P \) and \( \psi(P) \) is a gradient for $\Psi$ at
  \( P \).
\end{proposition}

\begin{proof}
  Let \( \{P_{\epsilon}\} \) be a quadratic mean
  differentiable path from \( P \in \mathcal{P} \) with score
  function \( \dot{\ell} \). By the assumption that the derivative of
  \( R_2(P, P_{\epsilon}) \) is zero at \( \epsilon=0 \) and
  the definition of pathwise differentiability
  (equation~\eqref{eq:pathwise-diff-def}), the result follows
  if we can show that
  \begin{equation}
    \label{eq:13}
    \lim_{\epsilon \downarrow 0} \epsilon^{-1} (P-P_{\epsilon})[ \psi(P)] = P{[\psi(P) \dot{\ell}]}.
  \end{equation}
  Recall that we assume the model \( \mathcal{P} \) to be dominated by
  the $\sigma$-finite measure $\nu$. We may thus write
  \begin{equation}
    \label{eq:19}
    \begin{split}      
    \epsilon^{-1}(P-P_{\epsilon})[ \psi(P)]
    & =
      \nu{[\psi(P) \epsilon^{-1}(p - p_{\epsilon}) ]}
    \\
    & =
      \nu{[ \psi(P) p \dot{\ell}]}
      +
      \nu{[(\epsilon^{-1}(p - p_{\epsilon}) - p \dot{\ell}) \psi(P)]}
    \\
    & =
      P{[\psi(P) \dot{\ell}]}
      +
      \nu{[(\epsilon^{-1}(p - p_{\epsilon}) - p \dot{\ell}) \psi(P)]}.
    \end{split}
  \end{equation}
  By Hölder's inequality and the assumption about $\psi(P)$, we have
  that
  \begin{equation*}
    | \nu{[(\epsilon^{-1}(p - p_{\epsilon}) - p \dot{\ell}) \psi(P)]} |
    \leq \| \psi(P) \|_{\infty} \| \epsilon^{-1}(p -
    p_{\epsilon}) - p \dot{\ell} \|_{\nu, 1}.
  \end{equation*}
  Lemma~\ref{lemma:dqm-l1-conv} below shows that this last expression
  is \( \smallO{(1)} \) for $\epsilon \downarrow 0$. Hence,
  equation~\eqref{eq:19} can be written as
  \( \epsilon^{-1}(P-P_{\epsilon})[ \psi(P)] = P{[\psi(P) \dot{\ell}]} +
  \smallO{(1)} \) for $\epsilon \downarrow 0$, which shows
  equation~\eqref{eq:13}.
\end{proof}

The following lemma is based on arguments given in Appendix~A.5 of \citep{bickel1993efficient}.
\begin{lemma}
  \label{lemma:dqm-l1-conv}
  Let
  \( \{P_{\epsilon} : \epsilon > 0\} \subset \mathcal{P} \) be
  a quadratic mean differentiable path from
  \( P \in \mathcal{P} \) with score function \( \dot{\ell} \). Then
  \begin{equation*}
    \left\Vert
      \frac{p_{\epsilon} - p}{\epsilon} - p \dot{\ell}
    \right\Vert_{{\nu}, 1} \longrightarrow 0, \quad
    \text{for} \quad \epsilon \longrightarrow 0.
  \end{equation*}
\end{lemma}

\begin{proof}
  Firstly, it holds that \( p \dot{\ell} \in \mathcal{L}^1{({\nu})} \) by the
  Cauchy-Schwarz inequality, because
  \begin{equation*}
    \| p \dot{\ell} \|_{\nu, 1}
    \leq \|
    \sqrt{p} \|_{\nu, 2} \| \sqrt{p} \dot{\ell}
    \|_{\nu, 2} =  (\nu{[p]})^{1/2} (\nu{[p
      \dot{\ell}^2]})^{1/2}
    =
    1 (P{[\dot{\ell}^2]})^{1/2}
    =
    \| \dot{\ell} \|_{P, 2},
  \end{equation*}
  which is finite because \( \dot{\ell} \in \LP \). Next we write
  \begin{equation*}
    p_{\epsilon} - p = (\sqrt{p_{\epsilon}} - \sqrt{p})(\sqrt{p_{\epsilon}} + \sqrt{p})
    = 2 \sqrt{p}(\sqrt{p_{\epsilon}} - \sqrt{p}) - (\sqrt{p_{\epsilon}} - \sqrt{p})^2,
  \end{equation*}
 and so
  \begin{equation*}
    \frac{p_{\epsilon} - p}{\epsilon} - p \dot{\ell}
    =
    2 \sqrt{p}
    \left(
      \frac{\sqrt{p_{\epsilon}} - \sqrt{p}}{\epsilon}
      - \frac{1}{2} \sqrt{p}  \dot{\ell}
    \right) -
    \epsilon
    \left(
      \frac{\sqrt{p_{\epsilon}} - \sqrt{p}}{\epsilon}
    \right)^2.
  \end{equation*}
  Hence
  \begin{align*}
    \left\Vert \frac{p_{\epsilon} - p}{\epsilon} - p \dot{\ell}
    \right\Vert_{\nu, 1}
    & \leq
      2 \left\Vert
      \sqrt{p}
      \left(
      \frac{\sqrt{p_{\epsilon}} - \sqrt{p}}{\epsilon}
      - \frac{1}{2} \sqrt{p}  \dot{\ell}
      \right)
      \right\Vert_{\nu, 1}
    % \\
    % & \qquad
      +
      \epsilon \left\Vert
      \frac{\sqrt{p_{\epsilon}} - \sqrt{p}}{\epsilon}
      \right\Vert_{\nu, 2}^2
    \\
    & =
      2 \left\Vert
      \sqrt{p}
      \left(
      \frac{\sqrt{p_{\epsilon}} - \sqrt{p}}{\epsilon}
      - \frac{1}{2} \sqrt{p}  \dot{\ell}
      \right)
      \right\Vert_{\nu, 1}
      + \smallO{(\epsilon)},
  \end{align*}
  where the last equality follows from the assumption that the path is
  quadratic mean differentiable, as this states that
  \( \frac{\sqrt{p_{\epsilon}} - \sqrt{p}}{\epsilon} \) converges in
  \( \L{2}{\nu} \) to \( 1/2 \sqrt{p}\dot{\ell} \) which has finite
  \( \L{2}{\nu} \)-norm. The Cauchy-Schwarz inequality gives
  \begin{equation*}
    \left\Vert
      \sqrt{p}
      \left(
        \frac{\sqrt{p_{\epsilon}} - \sqrt{p}}{\epsilon}
        - \frac{1}{2} \sqrt{p}  \dot{\ell}
      \right)
    \right\Vert_{\nu, 1}
    \leq
    1 
    \left\Vert      
      \left(
        \frac{\sqrt{p_{\epsilon}} - \sqrt{p}}{\epsilon}
        - \frac{1}{2} \sqrt{p}  \dot{\ell}
      \right) 
    \right\Vert_{\nu, 2}
    = \smallO{(\epsilon)},
  \end{equation*}
  where the last equality follows because
  \( \{P_{\epsilon}\} \) is differentiable in quadratic mean
  at \( \epsilon=0 \).
\end{proof}

\section{Proofs}
\label{sec:proofs}

\subsection{Proof of Theorem~\ref{theorem:target-nondiff}}
\label{sec:proof-theorem-1}
  
By the smoothness assumption imposed on $\tau$ under the model
\( \mathcal{P}_S \), we can explicate the two cases assumed in the
theorem as follows.

\begin{enumerate}[label=(\alph*), topsep=0pt]
\item\label{item:1}
  There exists a \( j \in \{1, \dots, d\} \)
  and an open box \( B \subset \mathcal{W} \) such that
  \( \alpha \in \tau(B) \) and \( \tau_j'(w) \not = 0 \) for all \(
  w \in B \), where we define
  \begin{equation*}
    \tau_j' = \frac{\partial \tau}{\partial w_j}.
  \end{equation*}
\item\label{item:2}
  There exists a box
  \( B \subset \mathcal{W} \) such that $\tau(w)=\alpha$ for all
  \( w \in B \).
\end{enumerate}

We consider these two cases in turn. In case~\ref{item:1}, we show
that while $P \mapsto \gamma(P)(\alpha)$ is Gateaux differentiable, i.e., the limit
on the left-hand side of \eqref{eq:pathwise-diff-def} exists, this
limit cannot be written as an inner product with a fixed element of
\( \LP \). In case~\ref{item:2}, we show that the limit on the left-hand
side of \eqref{eq:pathwise-diff-def} does not exist.

\begin{enumerate}[label=(\alph*), topsep=0pt]
\item Let $g \colon \mathcal{W} \rightarrow \R$ be a smooth
  non-zero function with compact support contained in \( B \),
  e.g., a mollifier. Define a path
  \( \{P_{\epsilon} : \epsilon \in (0,\delta) \} \), for some
  $\delta>0$, by defining
  \begin{equation}
    \label{eq:9}
    \begin{split}
      \E_{P_{\epsilon}}{[Y \mid A=1, W=w]}
      & =
        \E_{P}{[Y \mid A=1, W=w]} + \epsilon g(w),
      \\
      \E_{P_{\epsilon}}{[Y \mid A=0, W=w]}
      & =
        \E_{P}{[Y \mid A=0, W=w]},
      \\
        P_{\epsilon}(\diff a, \diff w)
      & = P(\diff a, \diff w),
    \end{split}
  \end{equation}
  i.e., only the conditional outcome given treatment is modified. The
  CATE function of \( P_{\epsilon} \) is then
  \( \tau(P_{\epsilon}) = \tau + \epsilon g \). As \( g \) is smooth
  it follows that \( \{P_{\epsilon}\} \subset \mathcal{P}_S \). We use
  the shorthand notation $\tau_{\epsilon} = \tau(P_{\epsilon})$ and
  $\tau = \tau(P)$ in the following. We now write our target parameter
  as
    \begin{equation*}
      \gamma(P_{\epsilon})(\alpha) = P_{\epsilon}{(\tau_{\epsilon}(W) \leq
        \alpha)}
      = \int_{\mathcal{W}} \1{\{\tau_{\epsilon}(w) \leq \alpha
        \}}  P_{\epsilon}{(\diff w)}
      =
      \int_{\mathcal{W}} \1{\{\tau_{\epsilon}(w) \leq \alpha \}}  P{(\diff w)},
    \end{equation*}
    where we use that the marginal distribution of \( W \)
    stays the same along the path \( P_{\epsilon} \). Next we
    use that because $g$ has support in \( B \) we have
    \( \tau_{\epsilon} = \tau \) outside \( B \), and so
    \begin{equation}
      \label{eq:1}
      \begin{split}
        \gamma(P_{\epsilon})({\alpha})
        & =
          \int_{B} \1{\{\tau_{\epsilon}(w) \leq \alpha \}}
          P{(\diff w)}
          +
          \int_{B^c} \1{\{\tau(w) \leq \alpha \}}  P{(\diff w)}
        \\
        & =
          P(W \in B, \tau_{\epsilon}(W) \leq \alpha)
          +
          \int_{B^c} \1{\{\tau(w) \leq \alpha \}}  P{(\diff w)}.
      \end{split}
    \end{equation}
    Define
    \begin{equation*}
      \tau_{j,\epsilon}' = \frac{\partial \tau_{\epsilon}}{\partial w_j} 
    \end{equation*}
    Without loss of generality we assume that $\tau_j'>0$ on
    \( B \). Assume also, purely for notational convenience,
    that \( j=1 \) and write \( w = (w_j, w_{-j}) \) and
    \( B= B_j \times B_{-j} \). By picking $\delta$ small
    enough we can ensure that also
    \( \tau_{j,\epsilon}'(w) >0 \) for all \( w \in B \) and
    \( \epsilon \in (0, \delta) \). We can also pick $\delta$
    small enough to ensure that there is an open interval
    \( I_{\alpha} \subset \R \) such that
    \( \alpha \in I_{\alpha} \subset \tau_{\epsilon}(B) \) for all
    \( \epsilon \in (0, \delta) \). This implies that for all
    \( w_{-j} \in B_{-j} \) the function
    \( \tilde{w} \mapsto \tau_{\epsilon}(\tilde{w}, w_{-j}) \)
    is bijective on \( B_j \). Let
    $\tau_{j, \epsilon}^{-1}(\blank ; w_{-j})$ denote the
    inverse of this function, i.e.,
    \( \tau_{j, \epsilon}^{-1} \) is such that for all
    $\epsilon \in (0,\delta)$,
    \begin{equation}
      \label{eq:21}
      \tau_{\epsilon}(\tau_{j, \epsilon}^{-1}(z; w_{-j}),
      w_{-j})
      = z,
      \quad \text{for all }
      z \in  I_{\alpha},
      w_{-j} \in B_{-j}.
    \end{equation}
    We can now write
    \begin{align*}
      & P(W \in B, \tau_{\epsilon}(W) \leq \alpha)
      \\
      & = P(W_j \in B_j, W_{-j} \in B_{-j}, \tau_{\epsilon}(W_j, W_{-j}) \leq \alpha )
      \\
      & = \E_P{\left[
        \1{\{ W_{-j} \in B_{-j} \}}
        P(W_j \in B_j, \tau_{\epsilon}(W_j, W_{-j}) \leq \alpha \mid W_{-j})
        \right]}
      \\
      & = \E_P{\left[
        \1{\{ W_{-j} \in B_{-j} \}}
        P(W_j \in B_j, W_j \leq \tau_{j, \epsilon}^{-1}(\alpha; W_{-j}) \mid W_{-j})
        \right]}
      \\
      & = \E_P{\left[
        \1{\{ W_{-j} \in B_{-j} \}}
        P(l_j < W_j \leq \tau_{j, \epsilon}^{-1}(\alpha; W_{-j}) \mid W_{-j})
        \right]},
    \end{align*}
    where \( B_j = (l_j, r_j) \) and we use that we must have
    \( \tau_{j, \epsilon}^{-1}(\alpha; w_{-j}) \leq r_j \) for all
    \( w_{-j} \in B_{-j} \). Let \( F_j(\blank \mid w_{-j}) \)
    denote the conditional cumulative distribution function
    for \( W_j \) given \( W_{-j} = w_{-j} \), and let
    \( P_{-j} \) denote marginal distribution of \( W_{-j} \).
    Note that neither \( F_j \) nor \( P_{-j} \) depend on
    \( \epsilon \) because the distribution of \( W \) is left
    unchanged by $\epsilon$. With this notation we may then
    write
    \begin{equation}
      \label{eq:2}
      \begin{split}
        & P(W \in B, \tau_{\epsilon}(W) \leq \alpha)
        \\
        & = 
          \int_{B_{-j}} 
          \left\{
          F_j(\tau_{j, \epsilon}^{-1}(\alpha; w_{-j})
          \mid w_{-j})
          - F_j(l_j- \mid w_{-j})
          \right\}
          P_{-j}(\diff w_{-j})
      \end{split}
    \end{equation}     
    Equations~\eqref{eq:1} and~\eqref{eq:2} imply that
    \begin{equation}
      \label{eq:22}
      \frac{\partial }{\partial \epsilon}
      \Big|_{\epsilon=0}
      \gamma(P_{\epsilon})({\alpha})
      =
      \frac{\partial }{\partial \epsilon}      
      \Big|_{\epsilon=0}
      \int_{B_{-j}} 
      F_j(\tau_{j, \epsilon}^{-1}(\alpha; w_{-j}) \mid w_{-j})
      P_{-j}(\diff w_{-j}).
    \end{equation}
    %% TODO: Use a more elegant arugment here based on the implicit
    %% function theorem
    We now show that
    \( \epsilon \mapsto \tau_{j, \epsilon}^{-1}(\alpha; w_{-j}) \) is
    differentiable at $\epsilon=0$ for any \( w_{-j} \in B_{-j} \) and
    find the derivative. Let $\epsilon_n \downarrow 0$. For each
    $\epsilon_n$ there exists a (unique) \( w_n \) such that
    \( w_n = \tau_{j, \epsilon_n}^{-1}(\alpha; w_{-j}) \), and there
    is also a (unique) \( w_0 \) such that
    \( w_0 = \tau_{j}^{-1}(\alpha; w_{-j}) \), where we use
    \( \tau_{j}^{-1} \) to denote the inverse of
    \( \tau(\blank, w_{-j}) \). By the mean value theorem we have that
    \begin{align*}
      \tau_{\epsilon_n}(w_n, w_{-j})
      & = \tau(w_n, w_{-j}) + \epsilon_n g(w_n, w_{-j})
      \\
      & = \tau(w_0, w_{-j}) + (w_n-w_0) \tau_j'(\tilde{w}_n^{\tau},
        w_{-j})
      \\
      & \qquad
        +
        \epsilon_n 
        \left[
        g(w_0, w_{-j}) + (w_n-w_0) g_j'(\tilde{w}_n^{g}, w_{-j})
        \right]
    \end{align*}
    for some values
    \( \tilde{w}_n^{\tau}, \tilde{w}_n^{g} \in (w_0 \wedge
    w_n, w_0 \vee w_n)\), where we use \( g_j' \) to denote
    the partial derivative of $g$. It follows that
    \begin{align*}
      w_n 
      \left[
      \tau_j'(\tilde{w}_n^{\tau}, w_{-j})
      +
      \epsilon_n g_j'(\tilde{w}_n^{g}, w_{-j})
      \right]
      & = \tau_{\epsilon_n}(w_n, w_{-j}) -\tau(w_0, w_{-j}) +
        w_0 \tau_j'(\tilde{w}_n^{\tau}, w_{-j})
      \\
      & \qquad
        -
        \epsilon_n 
        \left[
        g(w_0, w_{-j}) -w_0 g_j'(\tilde{w}_n^{g}, w_{-j})
        \right]
      \\
      & = 
        w_0 
        \left[
        \tau_j'(\tilde{w}_n^{\tau}, w_{-j})
        +
        \epsilon_n g_j'(\tilde{w}_n^{g}, w_{-j})
        \right]
      \\
      & \qquad
        -
        \epsilon_n 
        g(w_0, w_{-j})
        ,
    \end{align*}
    we the last equality follows because
    \( \tau_{\epsilon_n}(w_n, w_{-j}) = \tau(w_0, w_{-j})=\alpha
    \), and so
    \begin{equation}
      \label{eq:25}
      w_n = w_0
      -\epsilon_n 
      \frac{g(w_0, w_{-j})}{\tau_j'(\tilde{w}_n^{\tau}, w_{-j})
        +
        \epsilon_n g_j'(\tilde{w}_n^{g}, w_{-j})}.
    \end{equation}
    By construction,
    \( \tau_{j, \epsilon_n}^{-1}(\alpha; w_{-j}) = w_n \) and
    \( \tau_{j}^{-1}(\alpha; w_{-j}) = w_0 \), and so 
    equation~\eqref{eq:25} gives
    \begin{equation}
      \label{eq:23}
      \frac{\tau_{j}^{-1}(\alpha; w_{-j}) - \tau_{j,
          \epsilon_n}^{-1}(\alpha; w_{-j})}{\epsilon_n}
      = \frac{w_0 - w_n}{\epsilon_n}
      = \frac{g(w_0, w_{-j})}{\tau_j'(\tilde{w}_n^{\tau}, w_{-j})
        +
        \epsilon_n g_j'(\tilde{w}_n^{g}, w_{-j})}
    \end{equation}
    By equation~\eqref{eq:25} and the smoothness of $\tau$ and
    $g$ we have that $w_n \rightarrow w_0$ for
    \( n \rightarrow \infty \), and so also
    \( \tilde{w}^{\tau}_n \rightarrow w_0 \).
    Equation~\eqref{eq:23} hence implies that
    \begin{equation}
      \label{eq:24}
      \frac{\partial }{\partial \epsilon}
      \Big|_{\epsilon=0}
      \tau_{j, \epsilon}^{-1}(\alpha; w_{-j})
      =
      \frac{g(w_0, w_{-j})}{\tau_j'(w_0, w_{-j})}
      =\frac{g(\tau_{j}^{-1}(\alpha; w_{-j}), w_{-j})}{\tau_j'(\tau_{j}^{-1}(\alpha; w_{-j}), w_{-j})}.
    \end{equation}
    Finally, equation~\eqref{eq:23} and the dominated
    convergence theorem imply that we may write
    equation~\eqref{eq:22} as
    \begin{equation}
      \label{eq:46}
      \begin{split}
        \frac{\partial }{\partial \epsilon}
        \Big|_{\epsilon=0}
        \gamma(P_{\epsilon})({\alpha})
        & =
          \int_{B_{-j}} 
          f_j(\tau_{j, \epsilon}^{-1}(\alpha; w_{-j}) \mid w_{-j})
          \frac{\partial }{\partial \epsilon}
          \Big|_{\epsilon=0}
          \tau_{j, \epsilon}^{-1}(\alpha; w_{-j})
          P_{-j}(\diff w_{-j})
        \\
        & =
          \int_{B_{-j}} 
          f_j(\tau_{j, \epsilon}^{-1}(\alpha; w_{-j}) \mid w_{-j})
          \frac{g(\tau_{j}^{-1}(\alpha; w_{-j}), w_{-j})}{\tau_j'(\tau_{j}^{-1}(\alpha; w_{-j}), w_{-j})}
          P_{-j}(\diff w_{-j}),
      \end{split}
    \end{equation}
    where we use \( f_j \) to denote the conditional density
    corresponding to \( F_j \), and the last equality follows
    from equation~\eqref{eq:24}.

    Equation~\eqref{eq:46} shows that
    \( P \mapsto \gamma(P)(\alpha) \) is Gateaux differentiable in the
    sense that the limit on the left-hand side of
    \eqref{eq:pathwise-diff-def} exists. We now show that the stronger
    requirement of pathwise differentiability fails. For
    $\gamma({\alpha})$ to be pathwise differentiable, we would need
    (c.f., equation~\eqref{eq:pathwise-diff-def}) to find a function
    $\psi_P \in \L{2}{P}$, such that
    \begin{equation}
      \label{eq:26}
      \int_{B_{-j}} 
      f_j(\tau_{j, \epsilon}^{-1}(\alpha; w_{-j}) \mid w_{-j})
      \frac{g(\tau_{j}^{-1}(\alpha; w_{-j}), w_{-j})}{\tau_j'(\tau_{j}^{-1}(\alpha; w_{-j}), w_{-j})}
      P_{-j}(\diff w_{-j})
      =
      P{[\psi_P \dot{\ell}]},
    \end{equation}
    for all sufficiently smooth functions \( g \), where
    \( \dot{\ell} \) is the score function of the path
    \( P_{\epsilon} \). We claim that such a function cannot
    exists. To see this, first define the set
    \begin{equation*}
      \check{B} = \left\{ w \in B : \tau(w) = \alpha \right\},
    \end{equation*}
    which is the intersection of \( B \) and the $\alpha$-level set of
    $\tau$. The left-hand side of equation~\eqref{eq:26} depends on
    $g$ only through the values that $g$ take on \( \check{B} \),
    because
    \( (\tau_{j}^{-1}(\alpha; w_{-j}), w_{-j}) \in \check{B} \) by
    definition of \( \tau_{j}^{-1} \), c.f.,
    equation~\eqref{eq:21}. Hence, given any $g$, we can modify
    \( g \) arbitrarily on $B \setminus \check{B}$ and leave the
    left-hand side of equation~\eqref{eq:26} unchanged.  The score
    function \( \dot{\ell} \) of the path we defined in
    equation~\eqref{eq:9} is
    \begin{equation*}
      \dot{\ell}(a, y, w) =
      a
      \left(
        y \frac{g(w)}{\mu_P(1, w)} - (1-y) \frac{g(w)}{1-\mu_P(1, w)}
      \right),
    \end{equation*}
    which will clearly in general be affected if $g$ is modified on
    $B \setminus \check{B}$. Hence, if the right-hand side of
    equation~\eqref{eq:26} should remain unchanged under such
    modifications of $g$, this implies that the function $\psi_P$
    would need to be zero on $B \setminus \check{B}$. By the implicit
    function theorem, \( \check{B} \) is a manifold of dimension
    strictly lower than \( d \), and so $\psi_P$ would need to be
    Lebesgue-almost surely zero on \( B \). This would imply that the
    right-hand side of equation~\eqref{eq:26} would be zero for any
    \( g \) with domain contained in \( B \). However, this cannot be
    the case, because we may pick $g$ to be, e.g., equal to one on a
    small open neighborhood \( \mathcal{N} \subset B \) and zero
    outside \( B \). For this choice of \( g \), the left-hand side of
    equation~\eqref{eq:26} will be non-zero. We conclude that there
    cannot exist a function $\psi_P \in \LP$ such that
    equation~\eqref{eq:26} holds, and hence $\gamma({\alpha})$ is not
    pathwise differentiable in case~\ref{item:1}.

  \item Define the path \( P_{\epsilon} \) in the same manner
    as in case~\ref{item:1}, but now pick a \( g \) such that
    \( P(g(W) > 0) > 0 \). By the same line of reasoning as
    above we have that
    \begin{equation*}
      \gamma(P_{\epsilon})({\alpha})
      =
      P(W \in B, \tau_{\epsilon}(W) \leq \alpha)
      +
      \int_{B^c} \1{\{\tau(w) \leq \alpha \}}  P{(\diff w)}.
    \end{equation*}
    In the present case, we have that
    \( \tau_{\epsilon}(W) = \alpha + \epsilon g(W) \) when
    \( W \in B \), and so for all \( \epsilon >0 \) we have
    \begin{equation*}
      P(W \in B, \tau_{\epsilon}(W) \leq \alpha)
      = P(W \in B,   \epsilon g(W) \leq 0)
      = P(W \in B,   g(W) \leq 0).
    \end{equation*}
    Because \( g \) is zero outside \( B \) and assumed to be
    strictly positive on a set of positive probability, we
    must have that
    \begin{equation}
      \label{eq:4}
      P(W \in B,   g(W) \leq 0) < P(W \in B).
    \end{equation}
    This show that
    \( \gamma(P)({\alpha}) - \gamma(P_{\epsilon})({\alpha}) = c \) for
    some \( c < 0 \). This implies that
    \( \epsilon \mapsto \gamma(P_{\epsilon})({\alpha}) \) is
    discontinuous at $\epsilon=0$ when \ref{item:2} holds, and
    thus $\gamma({\alpha})$ cannot be pathwise differentiable at any
    \( P \) for which \ref{item:2} holds.
  \end{enumerate}

\subsection{Proof of Theorem~\ref{theorem:anti-target-diff}}
\label{sec:proof-theorem-2}

We start by establishing a von~Mises expansion of the parameter
\( \Gamma(\alpha) \). The following calculations are heuristic, but we
justify them formally under some conditions in
Lemma~\ref{lemma:gateaux} below.

First write
\begin{equation}
  \label{eq:29}
  \begin{split}  
    \Gamma(P)({\alpha})
    &= \int_{-1}^{\alpha}
      \E_P{[\1{\{\tau(P)(W) \leq u\}}]} \diff u
    \\
    &= 
      \E_P{ 
      \left[
      \int_{-1}^{\alpha} \1{\{\tau(P)(W) \leq u\}} \diff u
      \right]}
    \\
    &= 
      \E_P{ 
      \left[
      \int \1{\{\tau(P)(W) \leq u \leq \alpha\}} \diff u
      \right]}
    \\
    &= 
      \E_P{ 
      \left[
      \1{\{\tau(P)(W) \leq \alpha\}}(\alpha- \tau(P)(W))
      \right]} .
  \end{split}
\end{equation}
Following the heuristic approach outlined by, e.g.,
\cite{hines2022demystifying} and \cite{kennedy2022semiparametric}, we
now take the derivative of \( \Gamma(P_{\epsilon})({\alpha}) \) with
respect to $\epsilon$, where
\( P_{\epsilon}= P + \epsilon(\delta_{O_i} - P) \) and
\( \delta_{O_i} \) is the Dirac-measure with mass at \( O_i \). One
application of the product rule gives us that
  \begin{equation}
    \label{eq:5}
    \begin{split}
    \pd{\epsilon} \Gamma(P_{\epsilon})({\alpha})
    & =
      \E_{(\delta_{O_i}  - P)}{ 
      \left[
      \1{\{\tau(P)(W) \leq \alpha\}}(\alpha- \tau(P)(W))
      \right]}
    \\
    & \quad +
      \E_{P}{ 
      \left[
      \pd{\epsilon}
      \1{\{\tau(P_{\epsilon})(W) \leq \alpha\}}(\alpha- \tau(P_{\epsilon})(W))
      \right]}
    \\
    & =
      \1{\{\tau(P)(W_i) \leq \alpha\}}(\alpha- \tau(P)(W_i))
      - \Gamma_{\alpha}(P)
    \\
    & \quad +
      \E_{P}{ 
      \left[
      \pd{\epsilon}
      \1{\{\tau(P_{\epsilon})(W) \leq \alpha\}}(\alpha- \tau(P_{\epsilon})(W))
      \right]}.
    \end{split}
  \end{equation}
  A second heuristic application of the product rule gives us
  that
  \begin{equation}
    \label{eq:6}
    \begin{split}
      & \E_{P}{ 
        \left[
        \pd{\epsilon}
        \1{\{\tau(P_{\epsilon})(W) \leq \alpha\}}(\alpha- \tau(P_{\epsilon})(W))
        \right]}
      \\
      & =
        \E_{P}{ 
        \left[
        \pd{\epsilon}
        \1{\{\tau(P_{\epsilon})(W) \leq \alpha\}}(\alpha- \tau(P)(W))
        \right]}
      \\
      & \quad
        -
        \E_{P}{ 
        \left[    
        \1{\{\tau(P)(W) \leq \alpha\}}\pd{\epsilon} \tau(P_{\epsilon})(W)
        \right]}.
    \end{split}
  \end{equation}
  Whenever \( (\alpha - \tau(P)(W)) \not = 0 \) the function
  $\epsilon \mapsto \1{\{\tau( P_{\epsilon})(W) \leq \alpha\}}$
  is constant, and hence we would expect that
  \begin{equation}
    \label{eq:7}
    \E_{P}{ 
      \left[
        \pd{\epsilon}
        \1{\{\tau(P_{\epsilon})(W) \leq \alpha\}}(\alpha- \tau(P)(W))
      \right]} = 0.
  \end{equation}
  Next, it is a standard (heuristic) result
  \citep[e.g.,][]{hines2022demystifying} that
  \begin{equation*}
    \pd{\epsilon} \tau(P_{\epsilon})(w)
    = \frac{\delta_{W_i}(w)}{\rho(P)(w)}
    \left[      
      \frac{A_i}{\pi(P)(w)}(Y_i - \mu(P)(1, w))
      - \frac{1-A_i}{1-\pi(P)(w)}(Y_i - \mu(P)(0, w))
    \right],
  \end{equation*}
  where $\rho$ is the marginal density of \( W \), and
  $\pi$ and $\mu$ were defined in equations~\eqref{eq:10}
  and~\eqref{eq:11}, respectively. Hence
  \begin{equation}
    \label{eq:8}
    \begin{split}
      &\E_{P}{ 
        \left[    
        \1{\{\tau(P)(W) \leq \alpha\}}\pd{\epsilon} \tau(P_{\epsilon})(W)
        \right]}
      \\
      & =
        \1{\{\tau(P)(W_i) \leq \alpha\}}
        \left[      
        \frac{A_i}{\pi(P)(W_i)}(Y_i - \mu(P)(1, W_i))
        - \frac{1-A_i}{1-\pi(P)(W_i)}(Y_i - \mu(P)(0, W_i))
        \right].
    \end{split}
  \end{equation}
  The heuristic calculations in
  equations~\eqref{eq:5}-\eqref{eq:8} suggest that
  \begin{equation*}
    \pd{\epsilon} \Gamma(P_{\epsilon})({\alpha})
    = \upsilon_{\alpha}(P)
    % = \1{\{\tau(P)(W_i) \leq \alpha\}}(\alpha- \phi(O_i;P)) -
    % \Gamma_{\alpha}(P)
    ,
  \end{equation*}
  where $\upsilon_{\alpha}$ was defined in
  equation~\eqref{eq:anti-target-can-grad}. We formalize this with the
  expansion given in Lemma~\ref{lemma:gateaux} below.

  \begin{lemma}
  \label{lemma:gateaux}
  Assume that \( \mathcal{P} \) fulfills Assumption~\ref{assum1}. For
  any \( P', P \in \mathcal{P} \) we may write
  \begin{equation*}
    \Gamma(P')({\alpha}) - \Gamma(P)({\alpha})
    = (P'- P){[\upsilon_{\alpha}(P')]}
    + \mathrm{Rem}_{\alpha}(P', P),
  \end{equation*}
  where the function $\upsilon_{\alpha}$ was defined in
  equation~\eqref{eq:anti-target-can-grad}, and
  \begin{equation}
    \label{eq:lemma-gateaux-rem}
    \begin{split}
      \mathrm{Rem}_{\alpha}(P', P)
      = &
          P{
          \left[
          \frac{\eta_{\alpha}(P')}{\pi(P')}
          (\pi(P) - \pi(P'))
          (\mu(P)(1,\blank) - \mu(P')(1,\blank))
          \right]}
      \\
        & +
          P{
          \left[
          \frac{\eta_{\alpha}(P')}{1-\pi(P')}
          (\pi(P') - \pi(P))
          (\mu(P)(0,\blank) - \mu(P')(0,\blank))
          \right]}
      \\
        & + P{[(\eta_{\alpha}(P')-\eta_{\alpha}(P))(\alpha - \tau(P))]},
    \end{split}
  \end{equation}
  with $\eta_{\alpha}$ defined in equation~\eqref{eq:12}.
\end{lemma}

  \begin{proof}[Proof of Lemma~\ref{lemma:gateaux}] Define for
    any \( P \) and \( P' \) the remainder term
    \begin{equation}
      \label{eq:20}
      \mathrm{Rem}_{\alpha}(P', P) = \Gamma(P')({\alpha}) +
      P{[\upsilon_{\alpha}( P')]}  - \Gamma(P)({\alpha}).
    \end{equation}
    As \( P{[\upsilon_{\alpha}( P)]} =0 \) for any
    \( P \in \mathcal{P} \), the expansion
    \begin{equation*}
      \Gamma(P')({\alpha}) - \Gamma(P)({\alpha})
      = (P'-P){[\upsilon_{\alpha}( P')]} + \mathrm{Rem}_{\alpha}(P', P),
    \end{equation*}
    follows by the definition of \( \mathrm{Rem}_{\alpha}(P', P) \)
    given in equation~\eqref{eq:20}. We now show that
    \( \mathrm{Rem}_{\alpha}(P', P) \) is given as stated in
    equation~\eqref{eq:lemma-gateaux-rem} of
    Lemma~\ref{lemma:gateaux}. Using equations~\eqref{eq:29}
    and~\eqref{eq:20}, we see that
    \begin{align*}
      & \mathrm{Rem}_{\alpha}(P', P)
      \\
      & = \E_P{\left[
        \1{\{\tau(P')(W) \leq \alpha\}}
        \left\{
        \alpha - \tau(P')(W)
        \right\}
        \right]}
      \\
      & \quad -
        \E_P{\left[
        \1{\{\tau(P')(W) \leq \alpha\}}
        \left\{
        \frac{A}{\pi(P')(W)}(Y - \mu(P')(1, W))
        \right\}
        \right]}
      \\
      & \quad
        +
        \E_P{\left[
        \1{\{\tau(P')(W) \leq \alpha\}}
        \left\{
        \frac{1-A}{1-\pi(P')(W)}(Y - \mu(P')(0, W))
        \right\}
        \right]}
      \\
      & \quad
        - \E_P{\left[ \1{\{\tau(P)(W) \leq
        \alpha\}}(\alpha-\tau(P)(W)) \right]}
      \\
      & = \E_P{\left[
        \1{\{\tau(P')(W) \leq \alpha\}}
        \left\{
        \alpha - \tau(P')(W)
        \right\}
        \right]}
      \\
      & \quad -
        \E_P{\left[
        \1{\{\tau(P')(W) \leq \alpha\}}
        \left\{
        \frac{\pi(P)(W)}{\pi(P')(W)}(\mu({P})(1, W) - \mu(P')(1, W))
        \right\}
        \right]}
      \\
      & \quad +
        \E_P{\left[
        \1{\{\tau(P')(W) \leq \alpha\}}
        \left\{
        \frac{1-\pi(P)(W)}{1-\pi(P')(W)}(\mu({P})(0, W) - \mu(P')(0, W)) 
        \right\}
        \right]}            
      \\
      & \quad
        - \E_P{\left[ \1{\{\tau(P)(W) \leq \alpha\}}(\alpha-\tau(P)(W)) \right]},
    \end{align*}
    where the second equality follows by the tower property.  Using
    the definition of the nuisance parameter $\eta_{\alpha}$ (c.f.,
    equation~\eqref{eq:12}) we may write
    \begin{align*}
      & \mathrm{Rem}_{\alpha}(P', P) =
      \\& P{[\eta_{\alpha}(P')(\alpha - \tau(P')) -
      \eta_{\alpha}(P)(\alpha - \tau(P))]}
      \\
      & 
        - P{
        \left[
        \eta_{\alpha}(P')
        \left\{
        \frac{\pi(P)}{\pi(P')}(\mu({P})(1, \blank) - \mu(P')(1,
        \blank))
        \right\}
        \right]}
        +
        P{
        \left[
        \eta_{\alpha}(P')
        \left\{
        \frac{1-\pi(P)}{1-\pi(P')}(\mu({P})(0, \blank) - \mu(P')(0,
        \blank))
        \right\}
        \right]}.
    \end{align*}
    The first term on the right-hand side above can be written as
    \begin{align*}
      & P{[\eta_{\alpha}(P')(\alpha - \tau( P')) -
        \eta_{\alpha}(P)(\alpha - \tau( P))]}
      \\
      & = P{[\eta_{\alpha}(P')(\tau( P) - \tau(
        P'))]}
        +
        P{[(\eta_{\alpha}(P')-\eta_{\alpha}(P))(\alpha -
        \tau( P))]}
      \\
      & = P{[\eta_{\alpha}(P')
        [\mu(P)(1,\blank) - \mu(P')(1,\blank)]
        ]}
        -P{[\eta_{\alpha}(P')
        [\mu(P)(0,\blank) - \mu(P')(0,\blank)]
        ]}        
      \\
      & \quad +
        P{[(\eta_{\alpha}(P')-\eta_{\alpha}(P))(\alpha - \tau( P))]}.
    \end{align*}
    Hence
    \begin{align*}
      \mathrm{Rem}_{\alpha}(P', P)
      = &
          P{
          \left[
          \frac{\eta_{\alpha}(P')}{\pi(P')}
          (\pi(P) - \pi(P'))
          (\mu(P)(1,\blank) - \mu(P')(1,\blank))
          \right]}
      \\
        & +
          P{
          \left[
          \frac{\eta_{\alpha}(P')}{1-\pi(P')}
          (\pi(P') - \pi(P))
          (\mu(P)(0,\blank) - \mu(P')(0,\blank))
          \right]}
      \\
        & + P{[(\eta_{\alpha}(P')-\eta_{\alpha}(P))(\alpha - \tau( P))]}.
    \end{align*}
  \end{proof}

  Our proof of Theorem~\ref{theorem:anti-target-diff} relies on
  Lemmas~\ref{prop:expansion-gradient} and~\ref{lemma:gateaux} and the following technical lemma, which
  is an adaptation of Lemma~2 in
  \citep{vdLaanLuedtke2014optimalBepress}. 

\begin{lemma}
  \label{lemma:luedtke-laan}
  Let $\{P_{\epsilon} \} \subset \mathcal{P}$ be a quadratic mean
  differentiable path from \( P \in \mathcal{P} \) fulfilling
  Assumption~\ref{assum-path} and let \( r \geq 1 \). If
  $u \mapsto \gamma(P)(u)$ is continuous at $\alpha$ then
  \begin{equation*}
    \frac{\partial }{\partial \epsilon} \Big|_{\epsilon=0}
    P_{\epsilon}{[(\eta_{\alpha}(P) -
      \eta_{\alpha}(P_{\epsilon}))(\alpha - \tau(P_{\epsilon}))^r]}
    = 0.
  \end{equation*}
\end{lemma}

\begin{proof}
  Define the functions \( f_0(w) = \tau(P)(w) - \alpha\)
  and \( f_{\epsilon}(w) = \tau(P_{\epsilon})(w) - \alpha\).
  Then
  \begin{equation}
    \label{eq:15}
    \begin{split}
      \left\vert
      P_{\epsilon}{[(\eta_{\alpha}(P) -
      \eta_{\alpha}(P_{\epsilon}))(\alpha - \tau(P_{\epsilon}))^r]}
      \right\vert
      &
        = 
        \left\vert
        \int (\1{\{f_0 \leq 0 \}} - \1{\{f_{\epsilon} \leq 0 \}}) f_{\epsilon}^r  \diff P_{\epsilon}
        \right\vert
      \\
      &
        \leq     
        \int
        \left\vert
        (\1{\{f_0 \leq 0 \}} - \1{\{f_{\epsilon} \leq 0 \}})
        \right\vert
        |f_{\epsilon}|^r  \diff P_{\epsilon}
      \\
      &
        \leq
        2^{r-1}
        \int
        \left\vert
        (\1{\{f_0 \leq 0 \}} - \1{\{f_{\epsilon} \leq 0 \}})
        \right\vert
        |f_{\epsilon}|  \diff P_{\epsilon}
      \\
      &
        \leq
        2^{r-1}
        C
        \int
        \left\vert
        (\1{\{f_0 \leq 0 \}} - \1{\{f_{\epsilon} \leq 0 \}})
        \right\vert
        |f_{\epsilon}|  \diff P,
    \end{split}
  \end{equation}
  where the second-to-last inequality follows because
  \( |f_{\epsilon}|\leq 2 \), and the last inequality follows by
  Assumption~\ref{assum-path}~\ref{item:3}. By the mean value theorem
  and Assumption~\ref{assum-path}~\ref{item:4} we obtain the uniform
  bound
  \begin{equation}
    \label{eq:14}
    \sup_{w \in \mathcal{W}}| f_0(w) - f_{\epsilon}(w) | \leq C \epsilon,
  \end{equation}
  for some finite constant \( C \) and $\epsilon$ small
  enough. As
  \( \1{\{f_0(w) \leq 0 \}} - \1{\{f_{\epsilon}(w) \leq 0 \}}
  \) is non-zero if and only if either
  \( f_0(w) \leq 0 < f_{\epsilon}(w) \) or
  \( f_{\epsilon}(w) \leq 0 < f_{0}(w) \), the uniform bound in
  equation~\eqref{eq:14} implies that when
  \( \1{\{f_0(w) \leq 0 \}} - \1{\{f_{\epsilon}(w) \leq 0 \}}
  \) is non-zero, we must have that
  \( |f_0(w)| < C \epsilon \). Hence,
  \begin{equation}
    \label{eq:16}
    \int
    \left\vert
      (\1{\{f_0 \leq 0 \}} - \1{\{f_{\epsilon} \leq 0 \}})
    \right\vert
    |f_{\epsilon}|  \diff P
    \leq
    \int
    \1{\{|f_0| \leq C \epsilon \}}
    |f_{\epsilon}|  \diff P,
  \end{equation}
  and applying the uniform bound from equation~\eqref{eq:14}
  again, we obtain
  \begin{equation}
    \label{eq:17}
    \int
    \1{\{|f_0| \leq C \epsilon \}}
    |f_{\epsilon}|  \diff P
    \leq
    \int
    \1{\{|f_0| \leq C \epsilon \}}
    |f_0 + C \epsilon|  \diff P
    \leq
    2C\epsilon \int
    \1{\{|f_0| \leq C \epsilon \}}   \diff P
  \end{equation}
  Equations~\eqref{eq:15},~\eqref{eq:16}, and~\eqref{eq:17}
  thus give that
  \begin{equation}
    \label{eq:18}
    \left\vert
      P_{\epsilon}{[(\eta_{\alpha}(P) -
        \eta_{\alpha}(P_{\epsilon}))(\alpha - \tau(        P_{\epsilon}))^r]}
    \right\vert
    = \bigO{
      \left(
        \epsilon P{(|f_0(W)| \leq C \epsilon)}
      \right)}.
  \end{equation}
  Finally, using the definition of \( f_0 \) we have that
  \begin{align*}
    P{(|f_0(W)| \leq C \epsilon)}
    = P{(f_0(W) \leq C
    \epsilon)} - P{(f_0(W) < - C \epsilon)}
    = 
    \gamma(P)(\alpha+C\epsilon) - \gamma(P)(\{\alpha - C\epsilon\}-).
  \end{align*}
  Because \( u \mapsto \gamma(P)(u) \) is a cumulative distribution
  function it is right-continuous so
  \( \gamma(P)({\alpha+C\epsilon}) \rightarrow \gamma(P)({\alpha}) \),
  when $\epsilon\downarrow 0$. Also, because \( \gamma(P) \) has
  left-hand limits, the function \( u \mapsto \gamma(P)(u-) \) is
  left-continuous
  %% Proof is something like:
  %% 
  % Assume \( f \) has left-hand limits everywhere. Let
  % \( u_n \uparrow u \) and for each \( n \) pick \( u_n - n^{-1} < t_n < u_n \) such
  % that
  % \begin{equation*}
  %   | f(t_n) - f(u_n-) | < \epsilon,
  % \end{equation*}
  % which is possible because the limit \( f(u_n-) \) are assumed to exist.
  % Then
  % \begin{equation*}
  %   | f(u-) -f(u_n-) | \leq
  %   | f(u-) - f(t_n) | + 
  %   | f(t_n) - f(u_n-) |  < | f(u-) - f(t_n) | + \epsilon
  %   = O(\epsilon),
  % \end{equation*}
  % where the last equality follows because \( t_n \uparrow u \) by
  % construction and hence \( f(t_n) \rightarrow f(u) \) because \( f \)
  % has left-hand limits.
  %% 
  so
  \( \gamma(P)({\{\alpha - C\epsilon\}-}) \rightarrow
  \gamma(P)({\alpha-}) \). Hence
  \begin{equation*}
    P{(|f_0(W)| \leq C \epsilon)} \longrightarrow
    \gamma(P)({\alpha}) - \gamma(P)({\alpha-}) = 0,
  \end{equation*}
  where the last equality follows because
  $ u \mapsto \gamma(P)(u)$ assumed continuous at $\alpha$.
  Thus we conclude from equation~\eqref{eq:18} that
  \begin{equation*}
        \left\vert
      P_{\epsilon}{[(\eta_{\alpha}(P) -
        \eta_{\alpha}(P_{\epsilon}))(\alpha - \tau(P_{\epsilon}))^r]}
    \right\vert = \smallO{(\epsilon)},
  \end{equation*}
  which proves the result.
\end{proof}

\begin{proof}[Proof of Theorem~\ref{theorem:anti-target-diff}]
  By Lemma~\ref{lemma:gateaux}, we may write
  \begin{equation*}
    \Gamma(P)({\alpha}) - \Gamma(P_{\epsilon})({\alpha})
    = (P- P_{\epsilon}){[\upsilon_{\alpha}(P)]}
    + \mathrm{Rem}_{\alpha}(P, P_{\epsilon}),
  \end{equation*}
  for any path \( \{P_{\epsilon} : \epsilon >0\} \), where
  \( \mathrm{Rem}_{\alpha}(P, P_{\epsilon}) \) was defined in
  equation~\eqref{eq:lemma-gateaux-rem}. By
  Assumption~\ref{assum1}, the function
  \( \upsilon_{\alpha}(P) \) is uniformly bounded and
  hence it follows from Proposition~\ref{prop:expansion-gradient}
  that \( \Gamma({\alpha}) \) is pathwise differentiable and
  that \( \upsilon_{\alpha} \) is a gradient if we can show
  that the derivative of
  \( \epsilon \mapsto \mathrm{Rem}_{\alpha}(P, P_{\epsilon})
  \) is zero at $\epsilon=0$. Because the model has a
  saturated tangent space and
  \( P{[\upsilon_{\alpha}(P)]} = 0 \) this will also
  show that \( \upsilon_{\alpha} \) is the unique influence
  function and hence the efficient influence function of
  $\Gamma({\alpha})$. To show that the derivative of the
  remainder term is zero, we write
  \begin{equation*}
    \mathrm{Rem}_{\alpha}(P, P_{\epsilon})
    = A_1(\epsilon) + A_0(\epsilon) + B(\epsilon)
  \end{equation*}
  where
  \begin{align*}        
    A_1(\epsilon)
    & =
      \nu{ 
      \left[ \rho(P_{\epsilon})
      \frac{\eta_{\alpha}(P)}{\pi(P)}
      (\pi(P_{\epsilon}) - \pi(P))
      (\mu(P_{\epsilon})(1,\blank) - \mu({P})(1,\blank))
      \right]},
    \\
    A_0(\epsilon)
    & =
      \nu{
      \left[ \rho(P_{\epsilon})
      \frac{\eta_{\alpha}(P)}{1-\pi(P)}
      (\pi(P) - \pi(P_{\epsilon}))
      (\mu(P_{\epsilon})(0,\blank) - \mu({P})(0,\blank))
      \right]},
    \\
    B(\epsilon)
    & = P_{\epsilon}{[(\eta_{\alpha}(P)-\eta_{\alpha}(P_{\epsilon}))(\alpha - \tau(P_\epsilon))]},
  \end{align*}
  where we recall that the model \( \mathcal{P} \) is assumed bounded
  by the $\sigma$-finite measure $\nu$. By the dominated convergence
  theorem and Assumption~\ref{assum-path}~\ref{item:4} we may write
  \begin{align*}
    & \frac{\partial }{\partial \epsilon} \Big|_{\epsilon=0}
      A_1(\epsilon)
    \\
    & =
      \int_{\mathcal{W}}
      \left[
      \frac{\partial }{\partial \epsilon}
      \Big|_{\epsilon=0}
      \rho(P_{\epsilon})(w)
      \frac{\eta_{\alpha}(P)(w)}{\pi(P)(w)}
      (\pi(P_{\epsilon})(w) - \pi(P)(w))
      (\mu(P_{\epsilon})(1,w) - \mu({P})(1,w))
      \right]
      \nu(\diff w).
  \end{align*}
  It then follows from the product rule that the right-hand side is
  zero. The same argument applies to \( A_0(\epsilon) \). Finally, it
  follows from Lemma~\ref{lemma:luedtke-laan} with \( r=1 \) that the
  derivative of \( B(\epsilon) \) at \( \epsilon=0 \) is zero because
  of Assumption~\ref{assum-path}~\ref{item:4} and the assumption that
  $u \mapsto \gamma(P)(u)$ is continuous at $\alpha$.
\end{proof}

\subsection{Proof of Theorem~\ref{theorem:gren-est-rates}}
\label{sec:proof-theor-refth}

To derive the asymptotic distribution of the cross-fitted
Grenander-type estimator \( \hat{\gamma}^{\text{Gr}}_n \) we use
Theorem~4 from \cite{westling2020unified}. To do so, we verify that
the conditions of that theorem are satisfied under the assumptions of
our Theorem~\ref{theorem:gren-est-rates}. We refer to the conditions
of Theorem~4 in \citep{westling2020unified} with the prefix WC. One
assumption is that \( \gamma(P) \) is differentiable in a neighborhood
around \( \alpha \), which we also state explicitly as an
assumption. Next, we note that we do not employ a domain
transformation \( \Phi \), which means that condition (WC.A4) is
void. It remains to verify that the cross-fitted one-step estimator
$\hat{\Gamma}_n^{\bullet}$ satisfies display (2) in
\citep{westling2020unified} and that conditions (WC.A5) and
(WC.B1)-(WC.B5) are fulfilled. In a moment, we turn to verify each of
these conditions in turn. We comment on the conditions along the way,
but as each condition is technical and long to state, we refer to
\cite{westling2020unified} for the precise statements.

We start by establishing a decomposition, which demonstrate that
display (2) in \citep{westling2020unified} holds, and which we refer
to repeatedly in the following. By adding and subtracting terms we
obtain for each \( k=1, \dots, K \) the expansion
\begin{equation}
  \label{eq:37}
  \begin{split}
    \hat{\Gamma}_n^{k}(\alpha) - \Gamma(P)(\alpha)
    & = \Gamma(\hat{P}_n^{-k}) +
      \mathbb{P}_n^k{[\upsilon_{\alpha}(\hat{P}_n^{-k})]} - \Gamma(P)(\alpha)
    \\
    & =
      \Gamma(\hat{P}_n^{-k}) +
      \mathbb{P}_n^k{[\upsilon_{\alpha}(\hat{P}_n^{-k})]} - \Gamma(P)(\alpha)
      \pm (\mathbb{P}_n^k-P){[\upsilon_{\alpha}(P)]}
      \pm P{[\upsilon_{\alpha}(\hat{P}_n^{-k})]}
    \\
    & = (\mathbb{P}_n^k-P){[\upsilon_{\alpha}(P)]}
    \\
    & \quad +
      (\mathbb{P}_n^k-P){[\upsilon_{\alpha}(\hat{P}_n^{-k})-\upsilon_{\alpha}(P)]}
    \\
    & \quad +
      \Gamma(\hat{P}_n^{-k}) +
      P{[\upsilon_{\alpha}(\hat{P}_n^{-k})]} - \Gamma(P)(\alpha)
    \\
    &= (\mathbb{P}_n^k-P){[\upsilon_{\alpha}(P)]}
      +
      (\mathbb{P}_n^k-P){[\upsilon_{\alpha}(\hat{P}_n^{-k})-\upsilon_{\alpha}(P)]}
      +
      \mathrm{Rem}_{\alpha}(\hat{P}_n^{-k}, P),
  \end{split}
\end{equation}
where we use the definition of \( \mathrm{Rem}_{\alpha} \) given in
equation~\eqref{eq:20}. This gives that
\begin{equation}
  \label{eq:WC-display2}
  \hat{\Gamma}_n^{\bullet}(\alpha) - \Gamma(P)(\alpha)
  = (\mathbb{P}_n-P){[\upsilon_{\alpha}(P)]}
  + H_{\alpha,n},
\end{equation}
where
\begin{equation}
  \label{eq:38}
  \begin{split}
    & H_{n}(\alpha)
    \\
    &=
    \frac{1}{K}\sum_{k=1}^{K}\underbrace{(\mathbb{P}_n^k -
    \mathbb{P}_n){[\upsilon_{\alpha}(P)]}}_{H_{n,1}^k(\alpha)}
    +
    \frac{1}{K}\sum_{k=1}^{K}
    \underbrace{(\mathbb{P}_n^k-P){[\upsilon_{\alpha}(\hat{P}_n^{-k})-\upsilon_{\alpha}(P)]}}_{H_{n,2}^k(\alpha)}
    +
    \frac{1}{K}\sum_{k=1}^{K}
    \mathrm{Rem}_{\alpha}\underbrace{(\hat{P}_n^{-k}, P)}_{H_{n,3}^k(\alpha)}.
  \end{split}
\end{equation}
Equation~\eqref{eq:WC-display2} is display~(2) in
\citep{westling2020unified} for the estimator
\( \hat{\Gamma}_n^{\bullet} \).

\paragraph{(WC.A5)} This condition state that
\( \| \hat{\Gamma}_n^{\bullet} - \Gamma(P) \|_{\infty} =
\smallO_P{(1)}\), where \( \|\blank\|_{\infty} \) denotes the supremum
norm over the interval \( [-1,1] \). We first note that because
\( \{ \nu_{\alpha}(P) : \alpha \in I\} \) is a Glivenko-Cantelli
class, it follows that
\( \|(\mathbb{P}_n-P){[\nu_{\alpha}]}\|_{\infty} = \smallO_P{(1)}\),
and so equation~\eqref{eq:WC-display2} implies condition (WC.A5) if we
can show that \( \| H_n \|_{\infty} =\smallO_P{(1)} \). We do this by
showing the \( \| H_{n,l}^k \|_{\infty} =\smallO_P{(1)} \) for each
\( l=1,2,3 \) and any \( k \).

We may write
\begin{equation*}
  \frac{1}{K} \sum_{k=1}^{K} H_{n, 1}^{k}(\alpha)
  =
  \frac{1}{K} \sum_{k=1}^{K} 
  \left(
    \frac{1}{n_k} - \frac{K}{n}
  \right) \sum_{i \not \in \mathcal{D}_k} \nu_{\alpha}(O_i)
  =
  \frac{1}{K} \sum_{k=1}^{K} 
  \left(
    \frac{n}{n_k} - K
  \right) \frac{1}{n} \sum_{i \not \in \mathcal{D}_k} \nu_{\alpha}(O_i),
\end{equation*}
where \( \mathcal{D}_k \) denotes the \( k \)'th fold of the data. By
Assumption~\ref{assum2}, that the partition of the data is made such
that \( n/n_k - K =  \smallO_P{(1)} \), so
that\begin{equation*} \frac{1}{K} \sum_{k=1}^{K} H_{n, 1}^{k}(\alpha)
  = \smallO_P{(1)} \frac{1}{K} \sum_{k=1}^{K} \sum_{i \not \in
    \mathcal{D}_k} \nu_{\alpha}(O_i) = \smallO_P{(1)}
  \mathbb{P}_n[\nu_{\alpha}].
\end{equation*}
By Assumption~\ref{assum1}, \( \nu_\alpha(P) \) is bounded so
\( |\mathbb{P}_n{[\nu_\alpha(P)]}| \leq C \), for some finite \( C \)
uniformly in \( \alpha \). Hence
\( \| \frac{1}{K} \sum_{k=1}^{K} H_{n, 1}^{k} \|_{\infty} =
\smallO_P{(1)} \). Next, conditional on \( \hat{P}_n^{-k} \), the
class of functions
\( \{ \nu_{\alpha}(\hat{P}_n^{-k}) - \nu_{\alpha}(P) : \alpha \in
[-1,1]\} \) is a Glivenko-Cantelli class which means that
\begin{equation*}
  \sup_{\alpha \in I} \left| 
    {(\mathbb{P}_n^k-P){[\upsilon_{\alpha}(\hat{P}_n^{-k})-\upsilon_{\alpha}(P)]}}
  \right|
  \arrow{P} 0,
\end{equation*}
so hence also
\( \| \frac{1}{K} \sum_{k=1}^{K} H_{n,2}^k\|_{\infty} \arrow{P} 0
\). To control the term \( H_{n,3}^k \), and for later reference, we
establish the following lemma, where we use the norm
\( \| f \|_P = (P{[f^2]})^{1/2} \).
\begin{lemma}  
  \label{lemma:rem-bound}
  Assume that \( \mathcal{P} \) fulfills
  Assumption~\ref{assum1}. Then for any
  \( P, P' \in \mathcal{P} \),
  \begin{align*}
    \mathrm{Rem}_{\alpha}(P',P)
    &  \leq
      P{[|\eta_{\alpha}({P'})  - \eta_{\alpha}({P})| \cdot | \tau(P') - \tau({P})|]}
      +
      \sqrt{2}c^{-1}(\|\pi(P') - \pi(P)\|_P \|\mu(P') -
      \mu(P)\|_P).
  \end{align*}
\end{lemma}

\begin{proof}[Proof of Lemma~\ref{lemma:rem-bound}]
  From the uniform bound on $\pi(P')$ and the Cauchy-Schwarz
  inequality we obtain
  \begin{align*}
    & \left|
      P{
        \left[
          \frac{\eta_{\alpha}(P')}{\pi(P')}
          (\pi(P) - \pi(P'))
          (\mu(P)(1,\blank) - \mu(P')(1,\blank))
        \right]}
      \right|
    \\
    & 
    \leq \frac{1}{c} \|\pi(P') - \pi(P)\|_P \|\mu(P')(1,
    \blank) -  \mu(P)(1, \blank) \|_P      ,
  \end{align*}
  and similarly for the case where \( a=0 \). This gives that
  \begin{align*}
    &
      \bigg|  P{
      \left[
      \frac{\eta_{\alpha}(P')}{\pi(P')}
      (\pi(P) - \pi(P'))
      (\mu(P)(1,\blank) - \mu(P')(1,\blank))
      \right]}
    \\
    & \quad
      + P{
      \left[
      \frac{\eta_{\alpha}(P')}{1-\pi(P')}
      (\pi(P) - \pi(P'))
      (\mu(P)(0,\blank) - \mu(P')(0,\blank))
      \right]} \bigg|
    \\
    & \leq
      \frac{1}{c} \|\pi(P') - \pi(P)\|_P
      \big\{
      \|\mu(P')(1, \blank) -  \mu(P)(1, \blank) \|_P
      +
      \|\mu(P')(0, \blank) -  \mu(P)(0, \blank) \|_P 
      \big\}
    \\
    & \leq
      \frac{\sqrt{2}}{c} \|\pi(P') - \pi(P)\|_P
      \|\mu(P') -  \mu(P) \|_P
  \end{align*}
  To bound the remainder term, it remains to bound
  \( P{[(\eta_{\alpha}(P')-\eta_{\alpha}(P))(\alpha -
    \tau( P))]} \). To do so, note that
  \( |\eta_{\alpha}(P')-\eta_{\alpha}(P)| \) is non-zero if
  and only if either
  \( \tau( P') \leq \alpha < \tau( P) \) or
  \( \tau( P) \leq \alpha < \tau( P') \). Hence,
  on the set where
  \( |\eta_{\alpha}(P')-\eta_{\alpha}(P)| \) is not zero, we
  must have that
  \( |\alpha - \tau( P)| \leq |\tau( P') -
  \tau( P)| \). This shows that
  \begin{equation*}
    |P{[(\eta_{\alpha}(P')-\eta_{\alpha}(P))(\alpha -
      \tau( P))]}|
    \leq
    P{[|\eta_{\alpha}({P'})  - \eta_{\alpha}({P})| \cdot | \tau(
      P') - \tau( P)|]}.
  \end{equation*}
  The statement then follows from the expression for the remainder
  term given by Lemma~\ref{lemma:gateaux}.
\end{proof}

By Lemma~\ref{lemma:rem-bound} and Assumptions~\ref{assum1} and~\ref{assum3} it follows that
\begin{equation*}
  H_{n,3}^k  = \bigO_P{(\| \hat{\mu}_n^{-k} - \mu(P) \|_P)} = \smallO_P{(1)},
\end{equation*}
where the last equality is Assumption~\ref{assump:gren2}. We conclude
that \( \| H_n \|_{\infty} = \smallO_P{(1)} \) as wanted.

\paragraph{(WC.B1-2)} Define the class of functions
\( \mathcal{G}_{\alpha, R} = \{g_{\alpha, u} : |u| \leq R \} \), for
\( g_{\alpha, u} = \upsilon_{\alpha+u}(P) -\upsilon_{\alpha}(P) \),
where we suppress the dependence on \( P \) for notational
convenience. Conditions (WC.B1) and (WC.B2) are about controlling the
complexity of the class of functions \( \mathcal{G}_{\alpha, R} \). In
our case, the elements \( g_{\alpha, u} \) can be written on the form
\begin{align*}
  g_{\alpha, u}(O)
  & =
  \1{\{\tau(W) \leq \alpha + u\}}(\alpha - \phi(O))
  \\
  & \quad
  + \1{\{\tau(W) \leq \alpha + u\}} u
  \\
  & \quad  
    - \1{\{\tau(W) \leq \alpha \}}(\alpha - \phi(O))
  \\
  & \quad      
    - (\Gamma(\alpha+u) - \Gamma(\alpha)),
\end{align*}
where \( \phi \) was defined in
Theorem~\ref{theorem:anti-target-diff}. Define the function classes
\( \tilde{\mathcal{G}} = \{\1{\{\tau \leq \alpha + u\}} : |u| \leq R
\} \) and
\( \mathcal{F} = \{ \1{\{ \blank \leq \alpha + u \} } : |u| \leq R \}
\) and note that this can be written as
\( \tilde{\mathcal{G}} = \mathcal{F} \circ \tau \). As
\( \mathcal{F} \) is a VC class, it follows from Lemma~2.6.18~(vii) in
\citep{van1996weak} that \( \tilde{\mathcal{G}} \) is a VC class. As
\( g_{\alpha, u} \) is made up of products and sums of bounded
functions from VC classes, it follows from Lemma~5.1 in
\citep{van2006estimating} that the uniform entropy of
\( \mathcal{G}_{\alpha, R} \) is bounded up to a constant by
\( -\log(\epsilon) \). This shows (WC.B1). To show (WC.B2), we write
\begin{align*}
  & g_{\alpha, u}(O)
  \\
  & =
  (\1{\{\tau(W) \leq \alpha + u\}} - \1{\{\tau(W) \leq \alpha \}})(\alpha - \phi(O))
  + \1{\{\tau(W) \leq \alpha + u\}} u
    - (\Gamma(\alpha+u) - \Gamma(\alpha)),
\end{align*}
and note that this is bounded by for all
\( g \in \mathcal{G}_{\alpha, R} \), by
\begin{equation*}
  \1{\{\alpha < \tau(W) \leq \alpha + R\}}(\alpha + \| \phi
  \|_{\infty})
  + R
  + \sup_{z \in (\alpha, \alpha+ R)}|\gamma(z)| R,
\end{equation*}
where the used the mean value theorem to bound the last term. By
Assumption~\ref{assum1}, the function $\phi$ is uniformly bounded and
as $\gamma$ assumed is differentiable in a neighborhood around
$\alpha$, it follows that for small enough \( R \), we can use
\begin{equation*}
  G_{\alpha, R} = \1{\{\alpha < \tau(W) \leq \alpha + R\}} c_1 + R c_2
\end{equation*}
as envelope function for \( \mathcal{G}_{\alpha, R} \), for suitable
constants \( c_1 \) and \( c_2 \). It follows that
\begin{align*}
  P[G_{\alpha, R}^2]
  &\leq 2c_1^2P({\alpha < \tau(W) \leq \alpha + R})
    + 2 R^2 c_2^2
    \\
  &= 2c_1^2P({\alpha < \tau(W) \leq \alpha + R}) + \bigO{(R^2)}
        \\
  &= 2c_1^2 (\gamma(\alpha+ R) - \gamma(\alpha)) + \bigO{(R^2)}
  \\
  &= 2c_1^2 (\gamma'(\tilde{\alpha}_R)R) + \bigO{(R^2)},
\end{align*}
for some \( \tilde{\alpha}_R \in (\alpha, \alpha+R) \) by the mean
value theorem and the assumption that $\gamma$ is differentiable in a
neighbourhood around $\alpha$. Hence
\( P[G_{\alpha, R}^2] = \bigO{(R)} \), which shows the first part of
(WC.B2). The second part of (WC.B2) follows as \( G_{\alpha, R} \) is
uniformly bounded.

\paragraph{(WC.B3)} This condition concerns decomposing the asymptotic
variance term \( \Sigma(s,t) = P{[\upsilon_t \upsilon_s ]} \) into
differentiable and non-differentiable components, where we again
suppress the dependence on \( P \) in $\upsilon_{\alpha}$ for
notational convenience. Expanding \( \Sigma(s,t) \) gives
\begin{align*}
  \Sigma(s,t)
  & = P[\1{\{\tau \leq s\}} \1{\{\tau \leq
    t\}}(s-\phi)(t-\phi)] 
  \\
  & \quad
    +\Gamma(t) \Gamma(s)
  \\
  & \quad
    - \Gamma(s)P[\1{\{\tau \leq t\}}(t-\phi)]
  \\
  & \quad
    - \Gamma(t)P[\1{\{\tau \leq s\}}(s-\phi)].
\end{align*}
By the tower property and the definition of \( \Gamma \),
\( P[\1{\{\tau \leq s\}}(s-\phi)] = \Gamma(s) \), and so
\begin{equation*}
  \Sigma(s,t)
  =\Sigma^*(s,t) + P[\1{\{\tau \leq s\}} \1{\{\tau \leq
    t\}}(s-\phi)(t-\phi)],
  \quad \text{where} \quad
  \Sigma^*(s,t) = - \Gamma(t) \Gamma(s).
\end{equation*}
Define \( V = \tau(W) \) and \( A(s,t,v) = \E{[(s-\phi)(t-\phi)
  \mid V=v]} \). Then
\begin{align*}
  P[\1{\{\tau \leq s\}} \1{\{\tau \leq t\}}(s-\phi)(t-\phi)]
  & =\E[\1{\{V \leq t \wedge s\}}{\E{[(s-\phi)(t-\phi)\mid V]}}]
  \\
  &  = \E{[{ \1{\{V \leq t \wedge s\}} [A(s,t,V)]}]}
  \\
  &  = \int_{-\infty}^{t \wedge s} A(s,t,v) \gamma(\diff v),
\end{align*}
where we use that the cumulative distribution of \( V \) is $\gamma$
by definition. Thus
\begin{equation*}
    \Sigma(s,t)
    =\Sigma^*(s,t) + \int_{-\infty}^{t \wedge s} A(s,t,v) \gamma(\diff v),
\end{equation*}
From this representation and the assumption that $\gamma$ is
differentiability in neighborhood around \( \alpha \) it is easily
seen that (WC.B3) holds.

\paragraph{(WC.B4-5)} These last conditions concern the locally
uniform control of the error term \( H_n \), which contains both an
empirical process term and a remainder term. Define the term
\begin{equation}
  \label{eq:big-K-term}
  \mathcal{K}_n(\delta;P) = n^{2/3}\sup_{|u|\leq \delta n^{-1/3}} | H_n(\alpha +
  u) - H_n(\alpha) |,
\end{equation}
where \( H_n \) was defined in
equation~\eqref{eq:38}. Condition~(WC.B5) can be replaced with
condition~(WC.A3) which is trivially satisfied in our case our because
\( H_n \) is uniformly bounded by Assumptions~\ref{assum1}
and~\ref{assum3}. Finally, condition~(WC.B4) is explicitly assumed to
hold by Assumption~\ref{assump:gren1}.

\begin{proof}[Proof of Theorem~\ref{theorem:gren-est-rates}]
  By the preceding arguments, the conditions of Theorem~4 in
  \citep{westling2020unified} is fulfilled, and hence it follows from
  this theorem that
  \begin{equation*}
    n^{1/3}(\hat{\gamma}_n^{\mathrm{Gr}}(\alpha) - \gamma(\alpha))
    \rightsquigarrow c(\alpha) Z,
  \end{equation*}
  where \( Z \) follows the Chernoff distribution and
  \( c(\alpha) = [4 \gamma'(\alpha) \kappa(\alpha)]^{1/3} \), with
  \( \kappa(\alpha) = A(\alpha, \alpha, \alpha) \gamma'(\alpha)
  \). We see that in our case, the scaling factor
  \( c(\alpha) \) simplifies to
  \( c(\alpha) = [4 \gamma'(\alpha)^2 A(\alpha, \alpha, \alpha) ]^{1/3}\). Finally, 
  \begin{equation*}
    A(\alpha, \alpha, \alpha) = \E{\left[ (\alpha - \phi)(\alpha -
        \phi) \mid \tau(W) = \alpha \right]}
    = \E{\left[ 
        \frac{\left( Y- \mu(A,W)   \right)^2}{\pi(W)^A (1-\pi(W))^{1-A}}  
        \midd \tau(W) = \alpha \right]},
  \end{equation*}
  so \( A(\alpha, \alpha, \alpha) = \sigma^2(\alpha) \) which proves
  the theorem.
\end{proof}

\subsection{Proof of Theorem~\ref{theorem:spline-coef} and Lemma~\ref{lemma:rate-local-sup-norm}}
\label{sec:proof-theorem-spline}

\begin{proof}[Proof of Theorem~\ref{theorem:spline-coef}]
Define for any \( -1 \leq a < b \leq 1 \) the parameter
\begin{equation*}
  \theta_a^b(P) = \int_a^b \Gamma(P)(u) \diff u,
\end{equation*}
We start by showing that \( \theta_a^b \) is pathwise differentiable
at \( P \) if $\gamma(P)$ is continuous at \( b \), and show that, in
this case, its efficient influence function is
\begin{equation*}
  \dot{\theta}_a^b = \int_a^b \upsilon_u \diff u,
\end{equation*}
with \( \upsilon_u \) defined in
equation~\eqref{eq:anti-target-can-grad}. The argument proceeds in the
same way as the proof of Theorem~\ref{theorem:anti-target-diff}. By
the von~Mises expansion given by Lemma~\ref{lemma:gateaux}, we have
that
\begin{equation*}
  \theta_a^b(P) - \theta_a^b(P_{\epsilon})
  = (P- P_{\epsilon}){
    [
    \dot{\theta}_a^b(P)
      ]}
  + \int_a^b \mathrm{Rem}_{u}(P, P_{\epsilon}) \diff u,
\end{equation*}
where \( \mathrm{Rem}_{u} \) is defined in
equation~\eqref{eq:lemma-gateaux-rem}. By
Proposition~\ref{prop:expansion-gradient}, to show that $\theta_a^b$
is pathwise differentiable with \( \dot{\theta}_a^b \) as efficient
influence function, we just need to show that the derivative of
\( \epsilon \mapsto \int_a^b \mathrm{Rem}_{u}(P, P_{\epsilon}) \diff u
\) is zero a $\epsilon=0$. We write
\begin{equation}
  \label{eq:33}
  \int_a^b \mathrm{Rem}_{u}(P, P_{\epsilon}) \diff u
  = A_1(\epsilon) + A_0(\epsilon) + B(\epsilon)
\end{equation}
with
\begin{equation}
  \label{eq:34}
  \begin{split}
    A_1(\epsilon)
    & = \int_a^b P_{\epsilon}{ \left[
      \frac{\eta_{u}(P)}{\pi(P)}
      (\pi(P_{\epsilon}) - \pi(P))
      (\mu(P_{\epsilon})(1,\blank) - \mu(P)(1,\blank))
      \right]} \diff u
    \\
    A_0(\epsilon)
    & = \int_a^b P_{\epsilon}{ \left[
      \frac{\eta_{u}(P)}{1-\pi(P)}
      (\pi(P_{\epsilon}) - \pi(P))
      (\mu(P_{\epsilon})(0,\blank) - \mu(P)(0,\blank))
      \right]} \diff u
    \\
    B(\epsilon)
    & = \int_a^b
      P_{\epsilon}{[(\eta_{u}(P)-\eta_{u}(P_{\epsilon}))(u
      - \tau(P_\epsilon))]}
      \diff u.
  \end{split}
\end{equation}
We may further decompose
\begin{equation}
  \label{eq:35}
  \begin{split}
    B(\epsilon)
    &=
      P_{\epsilon}{
      \left[
      \int_a^b (\eta_{u}(P)-\eta_{u}(P_{\epsilon}))(u
      - \tau(P_\epsilon))
      \diff u
      \right]}   
    \\
    & =
      \frac{1}{2} P_{\epsilon}{
      \left[ \eta_{b}(P)
      \left\{
      (b - \tau(P_{\epsilon}))^2
      - (\tau(P) - \tau(P_{\epsilon}))^2
      \right\}
      -
      \eta_{b}(P_{\epsilon})
      \left\{
      (b - \tau(P_{\epsilon}))^2
      \right\}
      \right]}   
    \\
    & =
      \underbrace{\frac{1}{2} P_{\epsilon}{
      \left[ 
      \left\{
      \eta_{b}(P) - \eta_{b}(P_{\epsilon})
      \right\}
      \left\{  b - \tau(P_{\epsilon})
      \right\}^2 \right]}}_{B_1(\epsilon)}
      -\underbrace{\frac{1}{2} P_{\epsilon}{
      \left[ \eta_{b}(P)
      \left\{  \tau(P) - \tau(P_{\epsilon})
      \right\}^2
      \right]}}_{B_2(\epsilon)}.
  \end{split}
\end{equation}

Under Assumption~\ref{assum-path}~\ref{item:4}, the dominated
convergence theorem implies that the derivatives of \( A_0 \),
\( A_1 \), and \( B_2 \) at $\epsilon=0$ are zero.
Lemma~\ref{lemma:luedtke-laan}, applied with \( r=2 \), imply that the
derivative of \( B_1 \) at \( \epsilon=0 \) is zero. This concludes
the proof that $\theta_a^b$ is pathwise differentiable with
$\dot{\theta}_a^b$ the efficient influence function, when \( \gamma \)
is continuous at \( b \).

Proposition~3.3.1 in \citep{bickel1993efficient} states that any
asymptotically linear estimator with influence function equal to the
efficient influence function is regular and locally efficient. Hence
to prove Theorem~\ref{theorem:spline-coef} we show that, under the
assumptions of the theorem, \( \hat{\Gamma}_n^{\bullet}(a) \) is
asymptotically linear with influence function $\upsilon_a$ for
\( a \in \{a_0, a_{L+1}\} \), and that
\( \int_a^b \hat{\Gamma}_n^{\bullet}(u) \diff u \) is asymptotically
linear with influence function \( \dot{\theta}_a^b \), for each
\( a, b \in \{a_0, \dots, a_{L+1}\} \). The theorem then follows
because the components of \( \zeta \) are linear combinations of
\( \Gamma(a_l) \) and \( \theta_{a_l}^{a_{l+1}} \),
\( l = 0, \dots, L \).

We recall that we defined the one-step estimator as
\begin{equation*}
  \hat{\Gamma}_n^{\bullet}(a) = \frac{1}{K}\sum_{k=1}^{K}
  \hat{\Gamma}_n^{k}(a),
  \quad \text{with} \quad
  \hat{\Gamma}_n^{k}(a)
  =
  \mathbb{P}_n^k{[\upsilon_a(\hat{P}_n^{-k})]},
\end{equation*}
so
\begin{equation*}
  \hat{\Gamma}_n^{\bullet}(a) - \Gamma(P) =
  \frac{1}{K}\sum_{k=1}^{K}
  \left\{
    \hat{\Gamma}_n^{k}(a) - \Gamma(P)
  \right\},
\end{equation*}
and
\begin{equation*}
  \int_a^b \hat{\Gamma}_n^{\bullet}(u) \diff u - \theta_a^b(P) =
  \frac{1}{K}\sum_{k=1}^{K}
  \left\{
    \int_a^b \hat{\Gamma}_n^{k}(u) \diff u - \theta_a^b(P)
  \right\}.
\end{equation*}
Hence asymptotic linearity of the averaged estimators
\( \hat{\Gamma}_n^{\bullet}(a) \) and
\( \int_a^b \hat{\Gamma}_n^{\bullet}(u) \diff u \) will follow if we
can show
\begin{align}
  \label{eq:31}
  \hat{\Gamma}_n^{k}(a) - \Gamma(P)(a)
  & =
    (\mathbb{P}_n^k-P){[\upsilon_a(P)]} + \smallO_P{(n_k^{-1/2})},
    \intertext{and}
  \label{eq:32}
  \int_a^b \hat{\Gamma}_n^{k}(u) \diff u - \theta_a^b(P)
  &=
    (\mathbb{P}_n^k-P){[\dot{\theta}_a(P)]} + \smallO_P{(n_k^{-1/2})},
\end{align}
for any \( k=1, \dots, K \), where \( n_k \) is the number of
observations in the \( k \)'th fold.

\paragraph{Asymptotic linearity of $\hat{\Gamma}_n^{k}(a)$ for \( a
  \in \{a_0, a_{L+1}\} \).}
We first use the expansion from equation~\eqref{eq:37}, which states
that
\begin{equation*}
  \hat{\Gamma}_n^{k}(a) - \Gamma(P)(a) =
  (\mathbb{P}_n^k-P){[\upsilon_a(P)]}
      +
      (\mathbb{P}_n^k-P){[\upsilon_a(\hat{P}_n^{-k})-\upsilon_a(P)]}
      +
      \mathrm{Rem}_{a}(\hat{P}_n^{-k}, P).
    \end{equation*}
It is a standard result
\cite[e.g.,][Lemma~1]{kennedy2022semiparametric} that if
\( \| \upsilon_a(\hat{P}_n^{-k}) - \upsilon_a(P) \|_P \arrowP 0 \)
then
\( (\mathbb{P}_n^k-P){[\upsilon_a(\hat{P}_n^{-k})-\upsilon_a(P)]} =
\smallO_P{(n_k^{-1/2})} \). Assumption~\ref{assump:spline1} imply that
\( \| \upsilon_a(\hat{P}_n^{-k}) - \upsilon_a(P) \|_P \arrowP 0 \),
and so it follows that
\begin{equation*}
  \hat{\Gamma}_n^{k}(a) - \Gamma(P)(a) =
  (\mathbb{P}_n^k-P){[\upsilon_a(P)]}
  +
  \mathrm{Rem}_{a}(\hat{P}_n^{-k}, P)
  + \smallO_P{(n_k^{-1/2})}.
\end{equation*}
This equation implies equation~\eqref{eq:31}, because
Assumptions~\ref{assump:spline2} and~\ref{assump:spline3} together
with Lemma~\ref{lemma:rem-bound} and Assumptions~\ref{assum1}
and~\ref{assum3} imply that
\( \mathrm{Rem}_{a}(\hat{P}_n^{-k}, P) = \smallO_P{(n_k^{-1/2})}\)

\paragraph{Asymptotic linearity of
  $\int_a^b \hat{\Gamma}_n^{k}(u) \diff u$ for \( a, b \in \{a_0, \dots, a_{L+1}\} \).}

Using arguments similar to the ones above, we obtain the expansion
\begin{equation}
  \label{eq:36}
  \int_a^b \hat{\Gamma}_n^{k}(u) \diff u - \theta_a^b(P)
    = (\mathbb{P}_n^k-P){[\dot{\theta}_a^b(P)]}
    +
    (\mathbb{P}_n^k-P){[\dot{\theta}_a^b(\hat{P}_n^{-k}) - \dot{\theta}_a^b(P)]}
    +
    \int_a^b \mathrm{Rem}_{u}(\hat{P}_n^{-k}, P) \diff u,
\end{equation}
and we now argue, as above, that the two right-most terms on the
right-hand side are \( \smallO_P{(n^{-1/2})} \). We have that
\begin{equation*}
  \dot{\theta}_a^b(P) = \tilde{\theta}_a^b(P) -
  P{[\tilde{\theta}_a^b(P)]},
  \quad \text{where} \quad
  \tilde{\theta}_a^b(P) = \int_a^b \eta_u(P)(u - \phi(P)) \diff u,
\end{equation*}
where $\phi$ was defined in Theorem~\ref{theorem:anti-target-diff}.
Hence, to show that
\( \| \dot{\theta}_a^b(\hat{P}_n^{-k}) - \dot{\theta}_a^b(P) \|_P
\arrowP 0 \), it is enough to show that
\( \| \tilde{\theta}_a^b(\hat{P}_n^{-k}) - \tilde{\theta}_a^b(P) \|_P
\arrowP 0 \). We have that
\begin{equation*}
  \tilde{\theta}_a^b
  = \int_a^b \eta_u(u - \phi) \diff u
  = \eta_b \int_{a \vee \tau}^b (u - \phi) \diff u
  = \frac{\eta_b}{2}
  \left\{
    (b- \phi)^2  - ({a \vee \tau} - \phi)^2 
  \right\}  
\end{equation*}
so Assumption~\ref{assump:spline1} imply that
\( \| \tilde{\theta}_a^b(\hat{P}_n^{-k}) - \tilde{\theta}_a^b(P) \|_P
\arrowP 0 \). To show asymptotic linearity of
$\int_a^b \hat{\Gamma}_n^{k}(u) \diff u$, it now only remains to show
that
\( \int_a^b \mathrm{Rem}_{u}(\hat{P}_n^{-k}, P) \diff u \arrowP 0
\). For this we use the expansions given in
equations~\eqref{eq:33}-\eqref{eq:35} to see that
\begin{align*}
  \left| \int_a^b \mathrm{Rem}_{u}(\hat{P}_n^{-k}, P) \diff u
  \right|
  &  = \bigO_P{(\|\hat{\mu}_n - \mu(P) \|_P \|\hat{\pi}_n - \pi(P) \|_P
    )}
  \\
  & \quad
    +
    \bigO_P{(\|\hat{\tau}_n - \tau(P) \|_P^2)}
  \\
  & \quad
    +
    \bigO_P{(P{
    \left[ 
    \left\{
    \hat{\eta}_{b,n} - \eta_{b}(P)
    \right\}
    \left\{  b - \tau(P)
    \right\}^2 \right]})}.
\end{align*}
When \( \hat{\eta}_{b,n} - \eta_{b}(P) \) is non-zero we must have
that either \( \hat{\tau}_n \leq b < \tau(P) \) or \(  \tau(P) \leq b
< \hat{\tau}_n \), so almost surely
\begin{equation*}
  \left\{
    \hat{\eta}_{b,n} - \eta_{b}(P)
  \right\}
  \left\{  b - \tau(P)
  \right\}^2
  \leq
  \left\{
    \hat{\eta}_{b,n} - \eta_{b}(P)
  \right\}
  \left\{  \hat{\tau}_n - \tau(P)
  \right\}^2,
\end{equation*}
and thus
\begin{align*}
  \left| \int_a^b \mathrm{Rem}_{u}(\hat{P}_n^{-k}, P) \diff u
  \right|
  &  = \bigO_P{(\|\hat{\mu}_n - \mu(P) \|_P \|\hat{\pi}_n - \pi(P) \|_P
    )} 
    +
    \bigO_P{(\|\hat{\tau}_n - \tau(P) \|_P^2)}
  \\
    &  = \bigO_P{(\|\hat{\mu}_n - \mu(P) \|_P \|\hat{\pi}_n - \pi(P) \|_P
    )} 
    +
      \bigO_P{(\|\hat{\mu}_n - \mu(P) \|_P^2)}.
\end{align*}
Hence \( \int_a^b \mathrm{Rem}_{u}(\hat{P}_n^{-k}, P) \arrowP 0 \) by
Assumption~\ref{assump:spline2}, and so asymptotic linearity of
$\int_a^b \hat{\Gamma}_n^{k}(u) \diff u$ follows from
equation~\eqref{eq:36} and the previous arguments.
\end{proof}

\begin{proof}[Proof of Lemma~\ref{lemma:rate-local-sup-norm}]

  When \( \hat{\eta}_{\alpha,n} \not = \eta_{\alpha} \), $\hat{\tau}_n$ and $\tau$ are
  on different sides of $\alpha$ and thus in this case \( | \alpha -
  \tau | \leq |\hat{\tau}_n
  - \tau | \), i.e.,
  \begin{equation}
    \label{eq:inclusion}
    \{ \hat{\eta}_{\alpha,n} \not = \eta_{\alpha}  \}
    \subseteq \{  | \alpha - \tau|
    \leq | \tau - \hat{\tau}_n| \}
    \subseteq
    \{  | \alpha - \tau|
    \leq d_n \},
  \end{equation}
  with \( d_n = \| \tau - \hat{\tau}_n\|_{\infty} \).  This implies
  that
  \begin{align*}
    P{[\1{\{ \hat{\eta}_{\alpha,n} \not = \eta_{\alpha}  \}}]}
    \leq
      P{[\1{\{  | \alpha - \tau| \leq d_n \}}]}
    & =
      P( \tau(W) \in [d_n - \alpha, d_n+\alpha] \mid d_n)
      \\
    & = \gamma(d_n + \alpha) - \gamma((d_n-\alpha) -).
  \end{align*}
  The assumption that $\gamma$ is assumed differentiable in a
  neighborhood around $\alpha$ implies that there exists a constant
  finite constant \( C > 0 \) such that
  \begin{equation*}
    \gamma(d_n + \alpha) - \gamma((d_n-\alpha) -) \leq C d_n,
  \end{equation*}
  and so
  \begin{equation*}
    P{[\1{\{ \hat{\eta}_{\alpha,n} \not = \eta_{\alpha}  \}}]} \leq C d_n =  C \| \hat{\tau}_n - \tau \|_{\infty}
  \end{equation*}
  Hölder's inequality then implies that
  \begin{equation*}
    |P{[(\hat{\eta}_{\alpha,n} - \eta_{\alpha}) (\hat{\tau}_n - \tau) ]}|
    \leq
    P{[ \1{\{\hat{\eta}_{\alpha,n} \not = \eta_{\alpha}\}} |\hat{\tau}_n - \tau| ]}
    \leq
    \1{\{\hat{\eta}_{\alpha,n} \not = \eta_{\alpha}\}} \| \hat{\tau}_n - \tau
    \|_{\infty}
    = \bigO{(\| \hat{\tau}_n - \tau \|_{\infty}^2)}.
  \end{equation*}
\end{proof}

\section{Confidence bands for splines}
\label{sec:conf-bands-splin}

We here provide a detailed formula for a simultaneously valid
confidence band. The constructing is standard and sometimes referred
to as a \( t \)-sup confidence band
\citep[e.g.,][]{montiel2019simultaneous}, see also section 5.7 in
\citep{wasserman2006all}. We use the notation and definitions from
Section~\ref{sec:appr-estim}. Let
\( H(\alpha) = (H_0(\alpha), \dots, H_{L+1}(\alpha)) \) denote the
vector of basis functions evaluated at $\alpha$. We may then write
\begin{equation*}
  \hat{\gamma}^{\#}_n(\alpha) = H(\alpha)^{\top} M^{-1} \zeta(\hat{\Gamma}_n^{\bullet}).
\end{equation*}
By Theorem~\ref{theorem:spline-coef}, the asymptotic variance of \(
\sqrt{n}(\zeta(\hat{\Gamma}_n^{\bullet}) -\zeta) \)  is given by
\begin{equation*}
  \Sigma, \quad \text{where} \quad
  \Sigma_{j,r} = P{[\dot{\zeta}_j^{\top} \dot{\zeta}_r]},
  \quad j = 0, \dots L+1, r = 0, \dots L+1, 
\end{equation*}
and by the delta method, the asymptotic distribution of the
coefficient vector estimate
\( M^{-1} \zeta(\hat{\Gamma}_n^{\bullet}) \) is given by
\begin{equation*}
  \Theta=M^{-1} \Sigma M^{-\top}.
\end{equation*}
Define the estimator
\begin{equation*}
  \hat{\Theta}_n = M^{-1} \hat{\Sigma}_n M^{-\top},
  \quad \text{where} \quad
  \hat{\Sigma}_{j,r,n} = \empmeas{[\hat{\dot{\zeta}}_j^{\top} \hat{\dot{\zeta}}_r]},
\end{equation*}
and the \( \hat{\dot{\zeta}} \) denotes the plug-in estimator of
efficient influence function $\dot{\zeta}$, obtained by replacing the
nuisance parameters by their estimated counterparts. A pointwise
asymptotic variance estimate is then
\begin{equation*}
  \hat{\sigma}_n^2(\alpha)
  = H(\alpha)^{\top} \hat{\Theta}_n H(\alpha).
\end{equation*}
To obtain a asymptotically valid confidence band over
\( [\alpha_0, \alpha_{L+1}] \), we first sample a
large number \( B \) of centered multivariate Gaussians with variance
equal to the estimated asymptotic variance of the coefficient vector,
i.e., we obtain the samples
\begin{equation*}
  Z_b \sim \mathcal{N}(0, \hat{\Theta}_n), \quad b=1, \dots, B.
\end{equation*}
To construct a band, in practice we use a very fine equidistant grid
\( \{u_g\}_{g=1}^G \subset [\alpha_0, \alpha_{L+1}] \), define
\begin{equation*}
  T_b = \max_{g = 1, \dots, G} \frac{|H(u_g)^{\top} Z_b|}{\hat{\sigma}_n(u_g)},
\end{equation*}
and define the empirical \( (1-\alpha) \)-quantile of
\( \{T_b : b=1, \dots, B\} \) as \( \hat{c}_{1-\alpha} \). The
suggested \( t \)-sup confidence band is then
\begin{equation}
  \label{eq:spline-ci}
  \hat{\gamma}^{\#}_n(u)  \pm \hat{c}_{1-\alpha} \hat{\sigma}_n(u) ,
  \quad u \in \{u_1, \dots, u_G\},
\end{equation}
where \( \hat{\gamma}^{\#}_n(\alpha) \) was defined in equation~\eqref{eq:spline-approx-est}.

\section{Additional information about the numerical studies}
\label{sec:addit-inform-about}

\begin{figure}[h]
  \centerline{\includegraphics[width=1\linewidth]{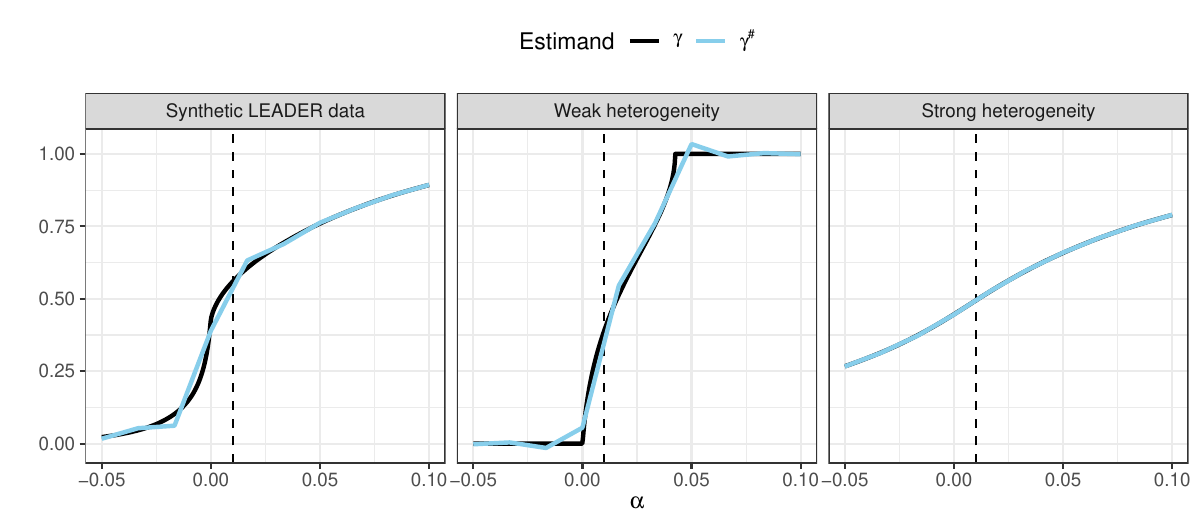}}
  \caption[]{The true sublevel function ($\gamma$) and its best
    piece-wise linear approximation ($\gamma^{\#}$) based on a fixed
    set of knot points in the interval \( [-0.05, .1] \) for each of
    the three data-generating mechanism described in
    Section~\ref{sec:numer-exper}. The dashed vertical line at
    $\alpha=0.01$ denotes the value at which we evaluated the
    pointwise performance of the suggested estimators.}
  \label{fig:sim-true-curves}
\end{figure}

\bibliography{ref.bib}

\end{document}